\documentclass[11pt,english]{article}

\usepackage[T1]{fontenc}
\usepackage[utf8]{inputenc}
\usepackage[english]{babel}
\usepackage[letterpaper,margin=1in]{geometry}

\usepackage{amsmath,amsthm,mathtools}
\usepackage{newtxtext}
\usepackage{newtxmath}
\usepackage{microtype}
\usepackage{array,booktabs}
\usepackage{float}
\usepackage{enumitem}
\usepackage{needspace}
\usepackage{setspace}
\usepackage{abstract}
\usepackage{titling}
\usepackage{titlesec}
\usepackage{etoolbox}
\usepackage[authoryear,round]{natbib}
\usepackage[unicode=true,pdfusetitle,
  bookmarks=true,bookmarksnumbered=false,bookmarksopen=false,
  breaklinks=true,colorlinks=false,pdfborder={0 0 0}]{hyperref}
\usepackage{tikz}
\usetikzlibrary{arrows,positioning,automata}

\allowdisplaybreaks
\setlist{leftmargin=2.1em,itemsep=0.2em,topsep=0.35em}
\titleformat{\section}
  {\large\bfseries}{\thesection.}{0.75em}{}
\titleformat{\subsection}
  {\normalsize\bfseries}{\thesubsection.}{0.75em}{}
\titleformat{\subsubsection}
  {\normalsize\itshape}{\thesubsubsection.}{0.75em}{}
\titlespacing*{\section}{0pt}{2.3ex plus 0.7ex minus 0.2ex}{1.1ex}
\titlespacing*{\subsection}{0pt}{1.7ex plus 0.5ex minus 0.2ex}{0.8ex}
\titlespacing*{\subsubsection}{0pt}{1.3ex plus 0.4ex minus 0.2ex}{0.5ex}

\pretitle{\begin{center}\Large\bfseries}
\posttitle{\par\end{center}\vspace{0.35em}}
\preauthor{\begin{center}\normalsize}
\postauthor{\par\end{center}}
\predate{\begin{center}\normalsize}
\postdate{\par\end{center}\vspace{0.6em}}

\AtBeginEnvironment{abstract}{\singlespacing}
\AtBeginEnvironment{verbatim}{\singlespacing\small}
\AtBeginEnvironment{thebibliography}{\singlespacing\fontsize{9}{9.8}\selectfont}

\providecommand{\assumptionname}{Assumption}

\providecommand{\claimname}{Claim}
\providecommand{\corollaryname}{Corollary}
\providecommand{\definitionname}{Definition}
\providecommand{\examplename}{Example}
\providecommand{\lemmaname}{Lemma}
\providecommand{\propositionname}{Proposition}
\providecommand{\remarkname}{Remark}
\providecommand{\theoremname}{Theorem}

\theoremstyle{plain}
\newtheorem{theorem}{\protect\theoremname}

\newtheorem{Proposition}{\protect\propositionname}
\newtheorem{proposition}[Proposition]{\protect\propositionname}

\newtheorem*{lem*}{\protect\lemmaname}

\newtheorem{corollary}{\protect\corollaryname}
\newtheorem{assumption}{\protect\assumptionname}

\theoremstyle{definition}

\theoremstyle{remark}

\newcommand{\W}{W}

\begin{document}
\begin{center}
  {\Large\bfseries The Attention Cost of Stable Matching\footnote{We thank Roy
  Allen, Francis Bloch, and Roberto Serrano for useful comments and encouragement.}}\\[0.55em]
  {\large Victor H. Aguiar and Dian Hong\footnote{Department of Economics,
  Simon Fraser University. Emails: vaguiarl@sfu.ca and
  dian\_hong@sfu.ca}}\\[0.25em]
  {\small This version: September 2026}
\end{center}
\vspace{0.35em}
\begin{abstract}
In large markets, scarce attention limits partner evaluation and creates
allocation loss, which stability magnifies. In an independent random market with
average executable degree \(d\), unmatched shares fall at rates
\(e^{-\sqrt d}\) under stability and \(e^{-d}\) under maximum matching on the
same graph. Changing consideration can make applications rejected in a
provisional active-screen computation relevant again. Exact query-neutral
implementation must retain allocation-relevant off-screen authorizations;
otherwise, missing authorization must be reacquired. Limited-attention deferred
acceptance (LA-DA) preserves valid authorizations and reengages eligible pairs.
Conditional on exact next-best information and persistent execution rights,
adaptive discovery saves a logarithmic factor in reached proposals relative to
independent exposure. In an application to speed dating, bilateral reports let us
compare stable and maximum matching on restricted graphs, separating missed
opportunities from same-graph stability loss. In Chilean school choice, we
document 9,502 applicants accepting higher-ranked or new placements through
retained rankings.


\end{abstract}

\section{Introduction}

Most matching theory begins after participants have inspected the market and
formed complete rankings. Many actual matching markets begin earlier: each
participant can evaluate only a small number of possible partners. This paper
studies the allocation cost of that inattention. Attention capacity constrains
which bilateral opportunities can become executable, and stability must then be achieved on the
resulting thin graph. We show that the no-blocking requirement magnifies the
allocation loss from scarce attention.

Limited attention is already an economic force in macroeconomics, industrial
organization, and organizational economics. Processing limits shape responses
to information and communication \citep{sims2003implications,dessein2016rational};
consideration capacity changes marketing, price competition, and platform
access \citep{eliaz2011consideration,declippeleliazrozen2014inattention,
cusumano2024competing,pratvalletti2022attention,teh2024accessibility,
chen2026attention}. Their allocation objects are consumer choice, trade,
communication, or directed-search efficiency. Ours is an exclusive two-sided
allocation subject to no blocking. In this setting, overlooking a partner does
not merely alter one person's choice. It removes an opportunity from both sides
and can change the set of stable assignments.

A dating application shows a short feed, a labor platform ranks a limited set
of applicants, and a school-choice system asks families to investigate only
some programs. Classical deferred acceptance begins after this attention
problem has been solved: participants already know and rank every relevant
partner \citep{gale1962college,roth2008deferred}. Modern platforms instead
create the opportunities that the mechanism later clears. A stable algorithm
cannot match a pair that never became possible. And when visible opportunities
change, forgetting an earlier valid application can destroy an opportunity the
platform already paid to create.

\Needspace{6\baselineskip}
The paper's central comparison holds fixed every opportunity created by scarce
attention:
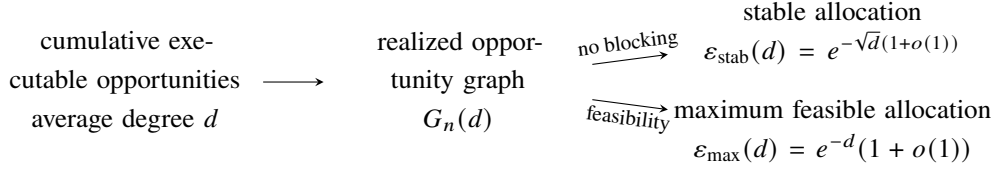
\begin{center}
\begin{tikzpicture}[
  >=stealth,
  every node/.style={align=center,font=\small,inner sep=2pt},
  attention/.style={text width=3.5cm},
  graph/.style={text width=3.4cm},
  result/.style={text width=4.2cm}
]
  \node[attention] (attention) {cumulative executable opportunities\\average degree \(d\)};
  \node[graph] (graph) [right=0.8cm of attention]
    {realized opportunity graph\\\(G_n(d)\)};
  \node[result] (stable) [right=1.0cm of graph,yshift=0.65cm]
    {stable allocation\\
     \(\varepsilon_{\mathrm{stab}}(d)=e^{-\sqrt d(1+o(1))}\)};
  \node[result] (maximum) [right=1.0cm of graph,yshift=-0.65cm]
    {maximum feasible allocation\\
     \(\varepsilon_{\max}(d)=e^{-d}(1+o(1))\)};
  \draw[->] (attention) -- (graph);
  \draw[->] (graph) -- node[above,sloped,font=\scriptsize] {no blocking} (stable);
  \draw[->] (graph) -- node[below,sloped,font=\scriptsize] {feasibility} (maximum);
\end{tikzpicture}
\end{center}
Both branches use the same participants, realized opportunities, and rankings.
The gap therefore measures the allocation cost of imposing stability, not a
difference in discovery or information.

The unit of analysis must also be precise. A profile may be delivered
without being processed; a processed profile may not produce an application;
and an authorized pair may not be selected. Our large-market results count
cumulative mutually acceptable opportunities. Our recommendation results count
pairwise queries or profile deliveries. Our memory result prices the state
needed to preserve a valid authorization without asking the user again. We do
not equate a fleeting impression with an executable edge.

Our first benchmark combines the stable-market limit of
\citet{arnosti2015short} with the maximum-matching limit of
\citet{bordenave2013matchings}. Consider a large balanced market in which each pair becomes
mutually acceptable independently and each person has \(d\) opportunities on
average. Once this graph is formed, give the platform every realized preference
and unlimited computation. Every ex post stable mechanism nevertheless leaves
approximately \(e^{-\sqrt d}\) of each side unmatched. A maximum matching on
the \emph{same graph} leaves only \(e^{-d}\) unmatched. Matching 99 percent
therefore requires about 21 opportunities per person under stability and 5
under feasibility. Because every stable matching has the same matched set,
changing the proposing side or the stable selection rule cannot remove the
loss. Theorem~\ref{thm:attention-cost-stability} is an impossibility for stable
clearing after homogeneous static attention has formed the graph, not a
criticism of deferred acceptance.

Limited attention creates a second problem during execution. Classical DA is a
one-pass procedure because rejection is final: as proposals arrive, a
receiver's held proposal only improves. A changing consideration set breaks
that monotonicity even when preferences are fixed. An edge that supported a
rejection can leave the active screen, making an earlier proposal relevant
again. Asking the proposer to rediscover and resubmit that proposal consumes
precisely the scarce attention the platform is trying to conserve.

Theorem~\ref{thm:authorization-memory} makes this argument exact. If \(k\)
off-screen authorizations can independently change the eventual DA assignment,
an exact query-neutral implementation needs \(2^k\) distinguishable states, or
at least \(k\) bits. A platform that forgets one pivotal authorization must
either return the wrong allocation in some history or ask the user again.
Persistent authorization is therefore an economic state variable, not merely
a database convenience.

Limited-attention deferred acceptance (LA-DA) preserves authorization while
respecting current eligibility. A remembered application can remain dormant;
it becomes executable again only when consideration returns. LA-DA updates
tentative holders online as the eligible graph changes, rather than treating
an earlier rejection as permanent. Theorem~\ref{thm:main} proves termination,
capacity compliance, and proposer-optimal stability and Pareto efficiency on
the final consideration graph.

Off-screen final settlement is a distinct extension, not a requirement for
online LA-DA. If the institution permits every valid retained pair to be
executed outside current consideration, the platform can clear the cumulative
ledger instead (Corollary~\ref{cor:offscreen-settlement}). This changes the
feasible graph. A cumulative-DA benchmark must respect the same execution
rights; it cannot bypass limited consideration by assigning an ineligible pair.

The first two results give a sharp division of labor:
\[
 \text{clearing cannot create a missing opportunity, but execution need not
 destroy a discovered one.}
\]
The remaining problem is to decide which opportunities to discover. In a
complete random-preference market with persistent execution rights and an exact
next-best preference oracle, revealing only the next proposal reached by deferred acceptance
reproduces the full-information stable matching with an expected
\(n\log n\) proposals. Homogeneous static exposure
requires order \(n\log^2 n\) edges to support a complete stable matching.
Theorem~\ref{thm:adaptive-attention-gap} combines classical proposal and
random-graph bounds to quantify this benchmark. The saving is in reached
proposals, conditional on access to correct next-best choices; learning those
choices is a separate information cost.

This is not a claim that the platform should solve once for a perfect global
plan. Choosing which receiver columns to process under heterogeneous attention
costs is NP-complete, with a weakly NP-complete isolated-pair restriction. This is a
budget-allocation lower bound, not hardness caused by DA: under cumulative
settlement, the sparse complete-order problem that is difficult for an
active-screen rule is order independent. The paper instead uses an adaptive
rule that looks for a concrete failure of the current stability certificate:
an omitted pair that would block the current matching.
Theorem~\ref{thm:endogenous-recommendation} proves that adding such pairs and
reclearing reaches a stable and efficient matching on the primitive graph
when the additions remain executable. Exact finite Bellman optimization and the budgeted planning boundary
are appendix results, not separate headline contributions.

We examine both losses in data. The speed-dating study of
\citet{fisman2006gender} records nearly complete bilateral ratings and approval
decisions within 21 finite markets. This allows us to hide opportunities and
replay clearing rules on a known graph. On the stricter graph of pairs who both
requested another date, proposal-directed discovery uses 1.14 executable
opportunities per proposer and matches 60.8 percent of the short side. Random
exposure with the same opportunity budget yields a 51.0 percent stable match
rate, while maximum matching on those same random edges reaches 56.2 percent.
The comparison separates poor discovery from the additional same-graph cost of
stability.

The broader graph that treats every observed date as acceptable provides the
direct finite-market counterpart to the large-market theorem. There the
proposal path matches the short side completely using 3.65 opportunities per
proposer. Random stable clearing at the same budget matches 85.9 percent,
although 95.3 percent is feasible. The replay is a structural counterfactual,
not a randomized intervention. We use realized encounter order to test its
most important behavioral restriction: approval decisions show no precise
average linear drift from the first to the last date, while ratings decline
modestly. We report both results and do not treat order invariance as observed
fact.

Chilean school choice supplies complementary evidence about persistent state.
In the 2024 national admissions process, families submitted rankings before an
initial centralized allocation and those rankings remained active through a
wait-list allocation. Among 117,409 applicants carried into that stage, 9,502
accepted either a higher-ranked school or a new placement from the original
ranking. The accepted-gain rate is similar when applicants linked in an
administrative family block are excluded. This is not a LA-DA treatment
effect or a welfare estimate over unlisted schools. It establishes that
earlier applications remain consequential after the matching state changes,
at the scale of an operating centralized market.

Our closest literatures study attention in product and platform markets,
connectivity in random stable markets, short lists, strategic search, limited
consideration, or signals before stable clearing
\citep{declippeleliazrozen2014inattention,teh2024accessibility,
kanoria2025competition,potukuchi2025unbalanced,agarwal2023stable,
kanoria2021facilitating,masatlioglu2012revealed,manzini2014stochastic,
ashlagi2020clearing,hatfield2005contracts,grenet2022preference}. Neither attention capacity nor
platform-controlled visibility is itself new. The paper's object is their
interaction with stable allocation: capacity determines a changing on-screen
graph, valid applications persist outside that graph, and reengagement must
respect eligibility without requiring a repeated application. This claim is
narrower than ``dynamic DA'' or ``recommendation before matching.''
Appendix~\ref{app:literature-design} gives the detailed comparison.

The online controller, authorization-state lower bound, optional ledger settlement,
separation theorem, and blocker-directed recommendation are checked in Lean.
Formal verification is useful because current display, retained authorization,
tentative holders, and terminal settlement are distinct states absent from
classical DA. The large-market probability theorem is proved analytically.

A short extension considers \emph{Limited-attention top trading cycles}
(LA-TTC). Retained comparisons establish each owner's best remaining object,
and valid permissions support the resulting exchanges. This preserves the
classical housing allocation when transfers wait for certification. The
extension illustrates the reach of the information and authorization design;
it does not establish a general attention-budget frontier across mechanisms.

Section~\ref{sec:attention-cost} proves the static-attention impossibility.
Section~\ref{sec:sequential-implementation} establishes the necessity and
sufficiency of persistent authorization. Section
\ref{sec:endogenous-recommendation} studies assignment-aware discovery.
Section~\ref{sec:platform-evidence} presents the speed-dating and Chilean
evidence. Section~\ref{sec:la-ttc-extension} gives the exchange extension.
The appendices contain proofs, optimization extensions, formal-
verification details, and supporting institutional evidence.

\section{Why Stable Clearing Cannot Repair Thin Attention}
\label{sec:attention-cost}

This section establishes an impossibility after attention has selected the
opportunity graph. Under homogeneous static attention, no choice of stable
clearing rule can recover the matches lost because the opportunity graph is
thin. The comparison with unrestricted matching on the same graph isolates
this stability constraint from ordinary scarcity of acceptable partners.
The stable-side formula is the balanced independent-preference case of
\citet[Theorem 11]{arnosti2015short}. We combine it with the maximum-matching
limit of \citet{bordenave2013matchings} and matched-set invariance to state a
common-graph benchmark for all stable clearing rules. This benchmark motivates
the dynamic implementation problem; it is not a new stable-market fixed point.

\subsection{Random cumulative consideration}

Fix \(d>0\). For each integer \(n\ge d\), let \(M_n\) and \(W_n\) contain
\(n\) agents. A static
homogeneous-attention policy places a pair \((m,w)\) in the cumulative
executable opportunity graph \(G_n(d)\)
independently with probability \(d/n\). Thus each agent has
asymptotic expected degree \(d\). Conditional on the graph, every
endpoint-edge incidence receives an independent uniform rank mark. Each
agent strictly ranks neighboring partners by these marks and prefers every
neighbor to remaining unmatched.

The graph can represent cumulative mutually acceptable consideration after a
recommendation horizon. It need not be interpreted as every profile briefly
displayed: an edge is present only when the pair is executable by the matching
mechanism. Section~\ref{sec:sequential-implementation} separates this
cumulative graph from the path by which its opportunities are discovered and
authorized. Thus this theorem prices the executability stage of the funnel,
not raw impressions, scrolling time, or unilateral profile delivery.

Let \(S_n(d)\) be the number of pairs in a stable matching of \(G_n(d)\).
Every stable matching has the same matched set, so this number is independent
of the stable selection. Let \(A_n(d)\) be the size of a maximum-cardinality
matching on the same realized graph.

An allocation rule is \emph{graph feasible} if it matches agents only along
edges of \(G_n(d)\). It is \emph{ex post stable} if, for every realization and
every realization of its own randomization, its output is stable with respect
to \(G_n(d)\) and the realized rankings. Write \(|\Phi_n|\) for the number of
pairs matched by a rule \(\Phi_n\). The rule may observe the entire realized
graph and all rankings, may use unlimited computation, and may randomize. The
impossibility below therefore does not arise from imperfect information or a
poor clearing algorithm after the graph has been formed.

\begin{theorem}
\label{thm:attention-cost-stability}
Let \(\Phi_n\) be any graph-feasible, ex post stable allocation rule, and let
\(\Psi_n\) be any graph-feasible allocation rule. For every fixed \(d>0\), as
\(n\to\infty\),
\begin{align*}
 \frac{|\Phi_n|}{n}
 &=\frac{S_n(d)}{n}
   \xrightarrow{p}1-\varepsilon_{\mathrm{stab}}(d),\\
 \frac{|\Psi_n|}{n}
 &\leq\frac{A_n(d)}{n}\quad\text{in every realization},
 &\frac{A_n(d)}{n}
 &\xrightarrow{p}1-\varepsilon_{\max}(d).
\end{align*}
Maximum matching attains the second bound. The limiting losses are
\begin{align}
 \varepsilon_{\mathrm{stab}}(d)&=e^{-t_d},
 &t_d^2&=d(1-e^{-t_d}),\quad t_d>0,\label{eq:stable-fixed-point}\\
 \varepsilon_{\max}(d)&=\max_{0\le s\le1}F_d(s),
 &F_d(s)&=e^{-d e^{-ds}}-1+e^{-ds}(1+ds).
 \label{eq:max-fixed-point}
\end{align}
Moreover, as \(d\to\infty\),
\begin{equation}
 \varepsilon_{\mathrm{stab}}(d)
   =e^{-\sqrt d(1+o(1))},
 \qquad
 \varepsilon_{\max}(d)=e^{-d}(1+o(1)).
 \label{eq:attention-loss-rates}
\end{equation}
Let \(d_{\mathrm{stab}}(\varepsilon)\) be the unique degree at which the
limiting stable unmatched share equals \(\varepsilon\), and let
\(d_{\max}(\varepsilon)=\inf\{d:\varepsilon_{\max}(d)\leq\varepsilon\}\).
Then
\begin{equation}
 d_{\mathrm{stab}}(\varepsilon)
   =\frac{\log^2(1/\varepsilon)}{1-\varepsilon},
 \qquad
 d_{\max}(\varepsilon)=\log(1/\varepsilon)(1+o(1)).
 \label{eq:attention-degree-requirements}
\end{equation}
where the second expression holds as \(\varepsilon\downarrow0\). In
particular, for any \(\varepsilon\in(0,1)\), if
\(d<d_{\mathrm{stab}}(\varepsilon)\), then every graph-feasible allocation
rule \(\Psi_n\) satisfies
\[
 \Pr\!\left\{\Psi_n\text{ returns a stable matching and }
        \frac{|\Psi_n|}{n}\geq 1-\varepsilon\right\}
 \longrightarrow 0.
\]
\end{theorem}

The stable formula is generated by two local threshold messages. Following a
random edge, let \(p\) be the probability that its proposer is willing to use
it after processing preferred alternatives, and let \(h\) be the corresponding
receiver probability. Poisson thinning and the uniform rank marks give
\[
 p=\frac{1-e^{-dh}}{dh},
 \qquad
 h=\frac{1-e^{-dp}}{dp}.
\]
Every solution is symmetric. Writing \(t=dp=dh\) yields
\eqref{eq:stable-fixed-point}; the root is unmatched with probability
\(e^{-t}\). Correlation decay on the marked Poisson tree transfers this local
calculation to the finite graph. Equation~\eqref{eq:max-fixed-point} is the
corresponding maximum-matching formula specialized from
\citet{bordenave2013matchings}. Appendix~\ref{app:attention-cost-proof}
supplies the complete argument.

The fixed point has a direct interpretation. The quantity \(t_d\) is the
effective number of opportunities that survive two-sided ranking pressure at
a typical endpoint. For large \(d\), almost every agent has graph neighbors,
so \(1-e^{-t_d}\) is close to one and the fixed point becomes
\(t_d^2\simeq d\). Stable clearing therefore extracts only order \(\sqrt d\)
effective depth from \(d\) raw opportunities. Maximum matching does not have
to respect these two-sided preference comparisons; at high degree, its leading
failure is the probability \(e^{-d}\) that an agent receives no opportunity at
all.

\subsection{Economic interpretation}

The final probability statement is the central impossibility. Once static
attention has generated the graph, replacing deferred acceptance with another
stable rule, changing the proposing side, or randomizing over stable outcomes
cannot cross the bound. Every stable matching has the same matched set.

Equation~\eqref{eq:attention-degree-requirements} explains the source and size
of the bound.
Maximum matching fails mainly because some agents have too few graph
opportunities. Stability adds a two-sided ranking requirement: a feasible
edge may be unusable because accepting it would create a chain of preferred
deviations elsewhere. That additional constraint changes the exponent from
\(d\) to \(\sqrt d\).

Define the finite-market stability tax by
\[
 \tau_n(d)=\frac{A_n(d)-S_n(d)}{n}.
\]
Because the two matchings are evaluated on the same realized graph,
\(\tau_n(d)\) contains no difference in attention budgets, acceptability, or
information. Theorem~\ref{thm:attention-cost-stability} implies
\[
 \tau_n(d)\xrightarrow{p}
 \varepsilon_{\mathrm{stab}}(d)-\varepsilon_{\max}(d),
 \qquad
 \tau(d)=e^{-\sqrt d(1+o(1))}.
\]
Equivalently, as the target loss \(\varepsilon\) becomes small, stable clearing
requires about \(\log(1/\varepsilon)\) times as many opportunities per person
as feasibility alone. This ratio, rather than the cardinal value of any one
match, is the paper's measure of the attention cost of stability.

\begin{table}[!htbp]
\centering
\caption{Asymptotic opportunity requirements}
\label{tab:attention-requirements}
\begin{tabular}{lrrr}
\hline\hline
Target matched share & \(\varepsilon\) & Stable degree & Feasible degree \\
\hline
90.0 percent & 0.100 & 5.89 & 2.30 \\
99.0 percent & 0.010 & 21.42 & 4.61 \\
99.9 percent & 0.001 & 47.76 & 6.91 \\
\hline\hline
\end{tabular}
\begin{minipage}{0.86\textwidth}
\footnotesize\emph{Notes:} Stable degree uses the exact expression in
\eqref{eq:attention-degree-requirements}. Feasible degree reports the
large-\(d\) approximation \(\log(1/\varepsilon)\).
\end{minipage}
\end{table}

The theorem prices an opportunity frontier. Maximum matching measures what
the realized graph can support; stable matching adds the requirement that no
pair prefer to leave together. Their difference prices stability in the same
primitive unit---cumulative executable opportunities per participant---rather
than in an imposed cardinal welfare scale.

The impossibility conditions on a graph formed by independent homogeneous
attention. Adaptive recommendation chooses which opportunities are discovered
and therefore changes the graph itself. This is the design margin studied
below: the platform must improve discovery because changing the stable
clearing rule cannot cross the frontier.

The benchmark isolates stability in a balanced one-to-one market where all
realized neighbors are acceptable, graph edges are independent, and rank marks
are independent across agents and edges. Popularity, correlated compatibility,
unequal market sides, and endogenous exit change the numerical frontier. The
finite-market comparison remains exact: stable and maximum matching use the
same realized opportunities. Section~\ref{sec:platform-evidence} performs that
comparison using observed rankings and two definitions of acceptability.

The order of limits is sequential: first \(n\to\infty\) at fixed \(d\), and
then \(d\to\infty\).

\section{LA-DA: Online Clearing with Persistent Authorization}
\label{sec:sequential-implementation}

Theorem~\ref{thm:attention-cost-stability} takes the executable graph as given.
This section asks how to maintain an assignment when consideration changes.
An authorization can remain valid without being currently executable.
LA-DA preserves it for reengagement when the pair becomes eligible again.

\subsection{Why rejection is no longer final}

Ordinary deferred acceptance relies on rejection finality. A receiver's holder
can only improve as proposals arrive, so a proposal rejected today can never
become relevant tomorrow. Non-monotone consideration breaks that logic even
when preferences never change.

Consider two proposers and two receivers with
\[
\begin{array}{c@{\qquad}c@{\qquad}c@{\qquad}c}
 m_1: w_1\succ w_2,&m_2:w_1\succ w_2,&
 w_1:m_2\succ m_1,&w_2:m_1\succ m_2.
\end{array}
\]
When both proposers apply to \(w_1\), she holds \(m_2\) and rejects \(m_1\);
\(m_1\) then applies to \(w_2\). Now suppose \(w_1\) leaves \(m_2\)'s
active display, a later recommendation returns \(w_1\) to \(m_1\)'s active set,
and \(m_2\) applies to \(w_2\). Because \((m_1,w_1)\) remains authorized, the
tentative stable assignment at that holding boundary is \((m_1,w_1)\) and
\((m_2,w_2)\). The earlier rejected application has become pivotal. Forgetting
it instead leaves \(m_2\) unmatched in that active state unless \(m_1\) is
asked to authorize \(w_1\) again. The point is not changing tastes. It is that
the set against which a rejection was final has changed.

\subsection{Environment and mechanism}

Let \(M\) and \(W\) be finite sets of proposers and receivers. A matching
\(\mu\) is a one-to-one partial assignment, written on both sides, with
\(\mu(m)=w\) if and only if \(\mu(w)=m\); write \(\mu(i)=\emptyset\) when
agent \(i\) is unmatched. Every agent has strict preferences over real
partners and \(\emptyset\). We write \(w\succ_m w'\) and \(m\succ_w m'\) for
these orders. A pair is \emph{mutually acceptable} when both endpoints prefer
one another to \(\emptyset\). For any mutually acceptable graph \(G\),
\(\operatorname{DA}_M(G)\) denotes proposer-proposing deferred acceptance on
the preference lists restricted to \(G\).

An edge becomes \emph{authorized} after the relevant recommendation and
application decisions. Let \(C_m^t\) denote the profiles currently displayed
to proposer \(m\), and let \(K_w^t\) be the platform's ledger of proposals to
receiver \(w\). Current consideration constrains which pairs can be assigned,
not merely which profiles appear on an interface. It is capacity constrained;
the ledger is cumulative within a run. At a holding boundary the active graph is
\[
 E_t=E(C^t,K^t)=\{(m,w):w\in C_m^t,\ m\in K_w^t,\text{ and the pair is
 mutually acceptable}\}.
\]

It is useful to distinguish three graphs. The \emph{primitive graph} \(A\)
contains every pair that could become a mutually acceptable opportunity. The
\emph{current graph} \(E_t\) contains the retained proposals that are eligible
at a particular holding boundary. The \emph{terminal retained-
authorization graph}
\begin{equation}
 G_T=\{(m,w):m\in K_w^T\text{ and }(m,w)\text{ is mutually acceptable}\}
 \label{eq:terminal-executable-graph}
\end{equation}
contains all valid authorizations accumulated by the end of the run, including
dormant pairs. Online LA-DA clears \(E_t\), not this larger ledger graph.
Clearing \(G_T\) is a separate extension requiring permission to execute
off-screen pairs. In the baseline, preferences and authorization validity are
fixed during a run; current eligibility changes with consideration. Revocation
or expiry requires a separate validity update and is not the same as dormancy.

A continuation is \emph{query neutral} when it requires no endpoint to process
a profile or submit or refresh an application. We assume that the platform
already holds the preference comparisons needed to clear the eligible graph,
from recorded rankings or valid prior elicitation. Reengagement reuses a valid
authorization only once the pair is eligible again. A binary application alone
need not reveal all required comparisons; eliciting one is an additional query.
Off-screen execution is not implied by retaining either comparisons or consent.

\begin{theorem}
\label{thm:authorization-memory}
For every \(k\), there is a finite market with \(k\) disjoint potential pairs
and \(2^k\) histories that share the same current on-screen graph and the same
future reactivation path, but differ in which off-screen pairs retain valid
authorization. Any deterministic, zero-error, query-neutral implementation of
proposer DA on the valid authorization graph must distinguish all \(2^k\)
histories. Equivalently, it requires at least \(k\) bits of authorization
state. If it pools two histories that differ on a pair, then under their common
continuation it must either return the wrong DA allocation in one history or
obtain that authorization again through an additional user-facing query.
\end{theorem}

\begin{proof}
Create \(k\) disjoint components. In component \(i\), proposer \(m_i\) and
receiver \(w_i\) rank one another above remaining unmatched and every cross-pair
is unacceptable. Before the common continuation, all pairs are off-screen and
history \(a\in\{0,1\}^k\) retains authorization for \((m_i,w_i)\) exactly when
\(a_i=1\). The continuation makes every pair eligible. Proposer DA then matches
component \(i\) exactly when \(a_i=1\).

An exact query-neutral implementation must therefore recover \(a\) from its
stored state. Its state encoding is injective on \(\{0,1\}^k\), so it has at
least \(2^k\) values. If two distinct vectors share a state, deterministic
execution produces the same output after the common continuation even though
proposer DA differs on at least one component. A new query can distinguish the
histories, but persistence would have avoided that repeated attention cost.
\end{proof}

The lower bound is attained in this family by one membership bit in the
retained-proposal record for each pair. More generally, a record can be safely
discarded once no admissible continuation can make it pivotal. Appendix
\ref{app:formal-verification} gives the continuation-equivalence formulation
and the corresponding Lean cardinality theorem.

We call the online mechanism \emph{Limited-attention deferred acceptance}
(LA-DA). It updates tentative assignments as consideration changes; it does
not promise that an irrevocable match will remain stable against arbitrary
future events. A serial pass processes a permutation \(\rho\) of receiver
columns. Each receiver is recommended to every proposer before the proposal
and holder update. The controller:
\begin{enumerate}
  \item recommends an eligible receiver and updates capacity-constrained current
  consideration;
  \item allows each proposer to apply to his best profitable active option;
  \item retains each valid application in \(K\) even when its profile later
  leaves the screen;
  \item recomputes tentative receiver holders by proposer deferred acceptance
  on \(E(C,K)\), reactivating retained applications when their profiles return;
  \item stops after no profitable considered proposal remains unrecorded
  and holders are current, returning that consideration-feasible assignment.
\end{enumerate}
Recommendation changes current opportunities, \(K\) preserves valid prior
actions, and reengagement restores them without resubmission when eligibility
returns. The complete pseudocode is in Appendix~\ref{app:formal-verification}.

For primitive input
\(\theta=(\succ,k,C^0,\rho)\), let \(\succ\) be the strict preference profile,
let \(1\leq k_m\leq |W|\), let \(|C_m^0|\leq k_m\), and let \(\rho\) be a
permutation of \(W\). Call such an input \emph{admissible}. The interleaved
procedure generates \(C^T(\theta)\) and \(K^T(\theta)\); neither is an
additional primitive. Identify \(C^T\) with its pair graph and set
\(F_T=A\cap C^T\), the mutually acceptable final consideration graph.

\begin{theorem}
\label{thm:main}
For every finite strict-preference market and every admissible primitive input
\(\theta\), online LA-DA terminates, respects consideration capacity, and returns
\begin{equation}
 \mu^T(\theta)=\operatorname{DA}_{M}\bigl(E_T(\theta)\bigr)
             =\operatorname{DA}_{M}\bigl(F_T(\theta)\bigr).
 \label{eq:la-da-online-allocation}
\end{equation}
Its assignment is feasible in final consideration, proposer-optimal stable
on \(F_T(\theta)\), and Pareto efficient among matchings feasible on that graph.
\end{theorem}

\paragraph{Proof idea.}
A further serial pass requires progress in the preceding pass; each closure
batch records a new pair. The finite controller therefore terminates.
Normalized holders give DA on \(E_T\).
Terminal exhaustion makes every omitted edge in \(F_T\) weakly worse for its
proposer than his assignment; restoring those edges leaves proposer DA
unchanged. Appendix~\ref{app:formal-verification} gives the full argument.

\paragraph{Optional off-screen settlement.}
\begin{corollary}
\label{cor:offscreen-settlement}
Suppose the institution additionally permits every valid retained pair to be
executed off screen and the required comparisons are recorded. Following
online LA-DA by query-neutral cumulative settlement returns
\begin{equation}
 \widehat\mu^T(\theta)=\operatorname{DA}_{M}\bigl(G_T(\theta)\bigr),
 \label{eq:panda-allocation}
\end{equation}
the proposer-optimal stable and Pareto-efficient assignment on \(G_T(\theta)\).
\end{corollary}

This extension changes execution rights, not the validity of the online proof.
The two assignments may differ. Cumulative DA is a valid alternative on the
same eligible graph; clearing the full ledger is not a feasible comparator
when it assigns an ineligible pair. Neither result ranks match counts across
different graphs. The capacity bound is proposer-side, with receiver
comparisons supplied under the stated information conditions.

\subsection{Attention--execution separation under a fixed query technology}

We now state the modular role of LA-DA precisely. Let \(\Omega\) be a finite
set of hidden market states with a specified prior, let the admissible binary
query set be finite with deterministic answers and nonnegative additive costs, and let
\(\mu^\star(\omega)\) be the target allocation in state \(\omega\), such as
full-information proposer DA. Assume the query language identifies the target:
any two states giving the same answers to every admissible query have the same
\(\mu^\star\). Querying the entire finite set is therefore a feasible zero-error
policy. A deterministic discovery policy adaptively chooses queries and
terminates at an implementation state \(s\). An execution
rule \(X\) maps that state into an allocation. The execution rule is
\emph{query neutral}, as above, when it is a deterministic function of the
recorded transcript and requires no additional user-facing query.

Let \(\Pi(\mu^\star)\) be the class of decision trees that identify
\(\mu^\star(\omega)\) for every \(\omega\in\Omega\) using the platform's
admissible queries, and define
\begin{equation}
 V_{\mathrm{att}}^*=\min_{\pi\in\Pi(\mu^\star)}
   \mathbb E\bigl[Q_\pi(\omega)\bigr],
\label{eq:target-attention-value}
\end{equation}
where \(Q_\pi\) is total query cost. Deterministic repeated queries supply no
new information. Removing them leaves finitely many trees, and target
identifiability makes the feasible class nonempty, so the minimum exists.
A terminal implementation state
\(s\) is \emph{target certified} when
\(X(s)=\mu^\star(\omega')\) for every hidden state \(\omega'\) consistent with
the transcript leading to \(s\).

\begin{theorem}
\label{thm:attention-execution-separation}
Fix a finite target-identification problem as above.
\begin{enumerate}[label=\textnormal{(\roman*)}]
  \item Every deterministic zero-error joint discovery--execution controller with a
  query-neutral execution layer has expected query cost at least
  \(V_{\mathrm{att}}^*\).

  \item If a discovery policy \(\pi^*\) attains \(V_{\mathrm{att}}^*\) and all of its
  terminal states are target certified for execution rule \(X\), then the
  composition \(X\circ\pi^*\) is zero error and has expected query cost
  exactly \(V_{\mathrm{att}}^*\).

  \item Online LA-DA instantiates \(X\) on its terminal eligible graph. If
  \[
    \operatorname{DA}_M(E_T(\omega))=\mu^\star(\omega),
  \]
  then attention-optimal discovery followed by online LA-DA is optimal among
  zero-error joint controllers using the same query language. Under the extra
  execution rights of Corollary~\ref{cor:offscreen-settlement}, the same claim
  holds for ledger settlement with \(G_T\) replacing \(E_T\).
\end{enumerate}
\end{theorem}

\begin{proof}
Take any joint controller and replace each terminal implementation state by
the allocation returned by its execution rule. This produces a pure decision
tree for \(\mu^\star\) with exactly the same internal query nodes on every hidden
state. Query neutrality therefore preserves pathwise and expected query cost,
so \eqref{eq:target-attention-value} gives part (i). Target certification gives
correctness of the composition in part (ii), while replacing terminal states
does not change its cost. For part (iii), Theorem~\ref{thm:main} identifies the
online output with \(\operatorname{DA}_M(E_T)\); the displayed certificate
identifies it with \(\mu^\star(\omega)\). Corollary
\ref{cor:offscreen-settlement} gives the analogous conclusion on \(G_T\)
under the additional execution rights.
\end{proof}

\begin{corollary}
\label{cor:terminal-graph-separation}
Fix preferences. Two online LA-DA runs with the same final eligible graph
\(E_T\), or the same final consideration graph \(F_T\), return the same
assignment. Under Corollary~\ref{cor:offscreen-settlement}, equality of
\(G_T\) instead suffices for the same extended settlement. These conclusions
do not require the same recommendation order or intermediate history.
\end{corollary}

Discovery must build an executable graph that identifies the target. Online
LA-DA clears the current graph without forgetting dormant authorizations;
the extension clears the ledger only with the additional execution rights.
Full stability on \(A\) requires the coverage and persistent-executability
conditions in Theorem~\ref{thm:endogenous-recommendation}.

Recommendation changes which opportunities are created; reengagement avoids
paying again for a prior action when eligibility returns. Neither turns a
dormant opportunity into a currently feasible match. Appendix
\ref{app:execution-rights-frontier} relates the attention frontier to these
distinct execution domains.

\section{Assignment-Aware Discovery}
\label{sec:endogenous-recommendation}

Random recommendation treats every absent edge symmetrically. Once a current
matching is available, this is wasteful: most omitted pairs cannot block that
matching because at least one endpoint would reject the deviation. The
platform should direct attention toward the missing part of the stability
certificate.

The section derives the discovery layer in two steps. A next-best path
establishes the logarithmic value of adaptation. A blocker-directed rule then
turns the same certificate logic into a finite operational policy. Both
benchmarks require reached opportunities to remain executable. They concern
the persistent-execution extension, not arbitrary nonmonotone eligibility.
Budgeted
global planning and the exact finite Bellman benchmark are reported in
Appendices \ref{app:budgeted-discovery} and
\ref{app:endogenous-recommendation-proofs}. A parameterized live controller
remains a replication-repository extension.

\subsection{The logarithmic value of an endogenous path}

Begin with the cleanest environment in which static and endogenous attention
can be compared. There are \(n\) proposers and \(n\) receivers, every pair is
acceptable, and all \(2n\) preference lists are independent uniform
permutations. A \emph{next-best recommendation policy} activates an unmatched
proposer and reveals only his favorite receiver who has not yet rejected him.
The resulting proposal and holding decision are retained before the next
activation, and reached pairs remain executable even off screen.
The policy has access to an exact next-best oracle and to receiver
rankings. Its cost below counts reached proposals, conditional on that
information service; it does not include the reports required to construct
the oracle.

\begin{theorem}
\label{thm:adaptive-attention-gap}
In the complete random-preference market with the information access above,
next-best recommendation followed by the optional ledger settlement in
Corollary~\ref{cor:offscreen-settlement} produces the
full-market proposer-optimal stable matching. If \(Q_n\)
is the number of recommended proposals and \(H_n=\sum_{j=1}^n1/j\), then
\begin{equation}
 \mathbb E[Q_n]\le nH_n
 \qquad\text{and}\qquad
 \mathbb E[Q_n]\sim n\log n.
 \label{eq:adaptive-proposal-cost}
\end{equation}
By contrast, suppose every pair is made executable independently before
clearing with probability \(p_n=c\log^2(n)/n\), where \(0<c<1\). With
probability approaching one, every stable matching on that graph leaves at
least
\begin{equation}
 \Delta_n=\frac{n^{1-\sqrt c}}{\log^2 n}
 \label{eq:static-stable-deficit}
\end{equation}
participants unmatched on each side. Thus homogeneous nonadaptive exposure
cannot produce a complete stable matching with high probability below order
\(n\log^2 n\). Conversely, if \(p_n=c\log^2(n)/n\) with \(c>9/4\), the
nonadaptive graph has a complete stable matching with probability approaching
one. The nonadaptive threshold therefore has order \(n\log^2 n\), whereas the
endogenous proposal path uses expected order \(n\log n\).
\end{theorem}

Appendix~\ref{app:endogenous-recommendation-proofs} gives the proof from
next-best sufficiency and the cited random-market bounds.

The comparison concerns complete stable assignment under two discovery rules.
Given correct next-best responses, adaptation saves a logarithmic factor in
reached proposals relative to independent static exposure. It does not bound
the cost of learning those responses or compare with every informed static
plan. In the frictionless benchmark the proposal path is ordinary asynchronous
DA. The optional settlement clears its recorded authorizations without additional
queries; the serial discovery stage need not be run again.

Exact ex ante planning is already combinatorial before matching interactions
are introduced. With heterogeneous delivery costs, choosing which receiver
columns to process under a budget is weakly NP-complete even in isolated
bilateral markets. This is Knapsack hardness, not a claim that a mandatory
complete receiver order is hard. Indeed, cumulative-ledger settlement removes
the order dependence in the sparse construction that arises under an
active-screen stopping rule. Appendix~\ref{app:budgeted-discovery} gives both
results and the tractable isolated-market boundary.

The policy below takes a different route. It conditions each recommendation
on the current assignment and targets a missing part of its stability
certificate. This bypasses ex ante subset optimization and guarantees a stable
endpoint, but it does not claim to minimize attention. Appendix
\ref{app:endogenous-recommendation-proofs} gives the exact finite adaptive
benchmark.

\subsection{A certificate-directed fallback}

Let \(A\subseteq M\times W\) be the full primitive graph of mutually acceptable
pairs, write \(A_m=\{w:(m,w)\in A\}\), and define \(G_m\) analogously for any
subgraph \(G\subseteq A\). Let \(G_0\subseteq A\) be the initially remembered
executable graph.
Receiver rankings are known. Proposer rankings may be directly reported or
certified from processed comparisons.

Appendix~\ref{app:endogenous-recommendation-proofs} defines the exact
potential-core policy and proves that it solves the finite hidden-compatibility
benchmark. The following policy replaces global Bellman search with a directly
observable progress measure: every query targets a demonstrated failure of the
current full-graph stability certificate.

Given current matching \(\mu=\operatorname{DA}_M(G)\), proposer \(m\)'s
willing-challenger menu
is
\begin{equation}
 R_m(\mu,G)=\{\mu(m)\}\cup
 \{w\in A_m\setminus G_m:m\succ_w\mu(w)\}.
 \label{eq:willing-challenger-menu}
\end{equation}
It contains the current assignment and exactly those omitted receivers who
would accept \(m\) in place of their current assignment.

At iteration \(t\), the policy:
\begin{enumerate}
  \item clears \(G_t\), obtaining
  \(\mu_t=\operatorname{DA}_M(G_t)\);
  \item certifies each proposer's favorite \(x_{mt}\) in
  \(R_m(\mu_t,G_t)\);
  \item recommends and retains \((m,x_{mt})\) whenever
  \(x_{mt}\neq\mu_t(m)\); and
  \item stops when every certified favorite is the current assignment.
\end{enumerate}
Proposer DA implements each reclearing. The optional extension permits the
growing remembered graph to remain executable; online LA-DA alone does not
make an off-screen pair eligible.

\begin{theorem}
\label{thm:endogenous-recommendation}
Suppose every favorite in \eqref{eq:willing-challenger-menu} is certified
correctly and every added edge remains executable in platform memory. Then:
\begin{enumerate}[label=\textnormal{(\roman*)}]
  \item every recommended pair blocks the current matching on the primitive
  graph;
  \item every nonterminal iteration adds at least one previously absent edge;
  \item the policy stops after at most \(|A\setminus G_0|\) nonterminal
  iterations; and
  \item its terminal matching is stable on \(A\) and Pareto efficient under
  the primitive preferences.
\end{enumerate}
\end{theorem}

The proof, including the graph-growth termination argument, is in
Appendix~\ref{app:endogenous-recommendation-proofs}.

The policy refines recommendation in three ways. First, it conditions on the
allocation, not only on pair scores. Second, it screens omitted candidates by
receiver willingness before spending proposer attention. Third, it stops on
an economic certificate---the absence of a primitive blocking pair---rather
than after an arbitrary number of profile deliveries.

The bound counts successful graph additions, not all comparisons or failed
queries needed to certify a favorite. It establishes finite attainment of full
stability, not an attention-optimal search rate.
The policy needs each reached proposer's favorite on the current
willing-challenger menu, not a complete ranking. Prediction can rank unexplored
pairs, but certification of this menu-specific choice determines when the
platform may stop. Appendix~\ref{sec:top-choice-identification} shows how
retained approvals and low-dimensional profile characteristics can identify
the choice without recovering the taste vector or unreached rankings.

\section{Applications: Attention and Persistent State}
\label{sec:platform-evidence}

The theory separates two losses that are usually bundled together. Discovery
can fail to create enough valuable bilateral opportunities, and a dynamic
clearing process can discard opportunities that participants have already
authorized. No single public dataset measures both margins well. We therefore
use two complementary settings. Nearly complete speed-dating reports let us
hold a finite market fixed, restrict the opportunities available to the
clearing rule, and compare stable and unconstrained assignments on exactly the
same graph. Chilean school-choice records let us observe whether applications
submitted before an initial DA allocation remain economically consequential
after the assignment state changes. The first exercise is a structural replay;
the second is descriptive evidence from an operating matching institution.

The execution rights differ from binding current consideration. The main
speed-dating replays keep each selected graph executable through clearing;
they measure discovery and the same-graph stability cost, not the performance
of online LA-DA under disappearing eligibility. The Chilean records document
retained applications and later reengagement, not unrestricted off-screen consent.

\subsection{Speed-dating markets and replay identification}

The speed-dating study of \citet{fisman2006gender} records 8,378 directed
reports from 21 experimentally organized events. At the end of each short date,
both participants rated the partner and independently decided whether they
wanted another date. Event sizes range from 5 to 22 participants on the
proposing side. We treat each event as a finite matching market, men as
proposers, and women as receivers. The analysis never moves a participant
across events. Requiring usable ratings in both directions leaves 4,011 pairs
among 277 proposers and 274 receivers; 686 pairs have mutual second-date
consent. Appendix~\ref{app:application-supplement} reports the sample flow and
the exclusion of an event with a cap on affirmative decisions.

For event \(e\), let \(\bar A_e\) contain the pairs with reports in both
directions, let \(R_{ij}\) be evaluator \(i\)'s rating of partner \(j\), and let
\(Y_{ij}\) indicate that \(i\) requests another date with \(j\). Like scores
produce strict rankings; attractiveness and a fixed identifier break ties.
We use two definitions of an executable opportunity. Our preferred,
decision-based graph retains \((i,j)\) only if \(Y_{ij}=Y_{ji}=1\). The broader
graph places every observed date above the outside option. The latter is the
direct finite-market counterpart to Theorem~\ref{thm:attention-cost-stability};
the former respects the participants' bilateral continuation decisions.

A replay policy may reveal, hide, or reorder pairs in \(\bar A_e\), but it
cannot create an unobserved pair. Interpreting the resulting allocations
requires the reports to be portable across counterfactual exposure paths.

\begin{assumption}
\label{ass:empirical-policy-invariance}
For every pair in \(\bar A_e\), the ordinal ranking implied by \(R_{ij}\) and
the acceptability decision \(Y_{ij}\) are invariant to the counterfactual
recommendation order and opportunity capacity. Participants have no relevant
partners outside event \(e\), and the policy's randomization, execution rights,
and transition rule are known.
\end{assumption}

The restriction is substantive. It rules out fatigue, sequence learning,
strategic reporting, and cross-event search as responses to the
counterfactual. Conditional on it, however, the unusually complete bilateral
reports identify more than a single allocation.

\begin{proposition}
\label{prop:complete-profile-replay}
Under Assumption~\ref{ass:empirical-policy-invariance}, the terminal
opportunity graph, online or extended LA-DA assignment, maximum-cardinality
match rate, and the specified rule's match and full-graph blocking-pair rates are identified for any
adaptive policy supported on \(\bar A_e\). For a stochastic policy, their
distribution is identified by integrating over the policy's known random seed.
\end{proposition}

The result follows recursively: after fixing a policy seed, each history
selects a pair whose ranking and acceptability response are already recorded.
The data therefore determine the next history and, eventually, every terminal
outcome. Appendix~\ref{app:application-supplement} gives the formal proof. This
is identification of a model-based policy replay, not identification of the
behavioral response to deploying that policy.

\subsection{Does encounter order change reports?}

The realized rotation lets us probe the assumption on the dimension observed
in the experiment. For directed report \((i,j,e)\), define \(P_{ije}\in[0,1]\)
as normalized encounter position, from the first to the last date. We estimate
separately for the second-date decision and the like score
\begin{equation}
 Z_{ije}=\alpha_i+\gamma_j+\rho P_{ije}+\varepsilon_{ije},
 \label{eq:order-invariance-diagnostic}
\end{equation}
where \(\alpha_i\) and \(\gamma_j\) are evaluator and partner fixed effects.
The coefficient \(\rho\) is the fitted first-to-last change within the realized
schedule. Standard errors are clustered across the 21 event-level markets.
The estimation sample contains 8,128 reports with complete outcomes and valid
order information.

\begin{table}[!htbp]
\centering
\caption{Average sensitivity of reports to encounter order}
\label{tab:order-invariance}
\small
\begin{tabular}{@{}llrrr@{}}
\hline\hline
Sample & Outcome & Mean & First-to-last change & 95\% interval \\
\hline
All & Second-date request & 0.43 & -1.40 pp & [-5.66,+2.87] pp \\
All & Like score (1--10) & 6.13 & -0.15 & [-0.28,-0.03] \\
Proposers & Second-date request & 0.49 & +1.16 pp & [-5.38,+7.69] pp \\
Proposers & Like score (1--10) & 6.28 & -0.12 & [-0.30,+0.06] \\
Receivers & Second-date request & 0.37 & -3.96 pp & [-9.88,+1.96] pp \\
Receivers & Like score (1--10) & 5.99 & -0.19 & [-0.37,-0.00] \\
\hline\hline
\end{tabular}
\begin{minipage}{0.94\textwidth}
\footnotesize\emph{Notes:} Each row regresses the directed report on normalized encounter position, with evaluator and partner fixed effects. The coefficient is the fitted change from the first to the last date. Intervals use wave-clustered standard errors and the $t_{20}$ critical value. The decision mean is a fraction; its change is reported in percentage points. These diagnostics test average sequence drift in the realized schedule. They cannot establish pair-level invariance under an arbitrary counterfactual order.
\end{minipage}
\end{table}

The approval decision that defines the mutual-consent graph changes by
\(-1.4\) percentage points from the first to the last encounter, with a 95
percent interval of \([-5.7,2.9]\). The corresponding like score falls by
0.15 on a ten-point scale. The proposer-side estimates are smaller and the
receiver-side estimates are more negative, but the approval intervals are
wide in both subsamples. A decline in report levels does not by itself identify
a change in latent partner rankings.

This diagnostic is informative but not decisive. The design does not observe
the same pair in two different positions, so it cannot test pair-level
invariance under every counterfactual order. We therefore retain
Assumption~\ref{ass:empirical-policy-invariance}. Appendix
Table~\ref{tab:rank-detrending-sensitivity} removes the fitted position
component while holding consent, edge budgets, and sampled graphs fixed.
The mutual-consent stability gap remains 4.62 percentage points, compared with
5.27 under the original rankings. All resulting rank reversals occur between
equal like scores. This supports the same-graph comparison under alternative
tie breaking; it does not validate order-invariant consent.

\subsection{From exposure to executable opportunity}

The distinction between a delivered profile and an executable edge matters for
interpreting the theory. The degree \(d\) in
Theorem~\ref{thm:attention-cost-stability} counts mutually acceptable
opportunities, whereas a platform typically pays first to deliver profiles.
For the persistent-execution replay below, suppose each pair is delivered at
most once and remains executable if accepted. Let \(b_m\) be proposer
\(m\)'s number of delivered profiles, let \(H_{mq}\) denote the history before delivery \(q\), and let
\(w_{mq}\) be the profile delivered then. The expected number of executable
opportunities created for \(m\) is
\begin{equation}
 d_m(\pi)=\mathbb E_{\pi}\!\left[\sum_{q=1}^{b_m}
 \Pr_{\pi}\!\left(Y_{mw_{mq}}=Y_{w_{mq}m}=1\mid H_{mq},w_{mq}\right)\right].
 \label{eq:delivery-to-opportunity}
\end{equation}
Raw delivery capacity \(b_m\), bilateral conversion, and the placement of the
resulting edges are therefore distinct design margins. Increasing delivery can
raise graph thickness; better targeting can create more executable edges from
the same number of deliveries; assignment-aware targeting can place those
edges where stable clearing needs them.

The two graph definitions make the conversion wedge visible without assigning
it a causal interpretation. At full observed exposure, the all-partners graph
has 12.55 edges per proposer on average, while requiring mutual second-date
consent leaves 2.19. Thus most observed meetings do not become bilateral
opportunities under the decision-based definition. The replay below conditions
on those realized decisions. It evaluates how a fixed collection of executable
opportunities is placed and cleared, but it does not estimate whether a new
recommendation policy would itself change attention, consent, or reporting.

\subsection{A common opportunity budget}

The first experiment asks how much the placement of a fixed number of
executable opportunities matters. For each graph definition, write
\(A_e\subseteq\bar A_e\) for the primitive executable graph. We compare three
allocations within each event.

First, we run proposer DA on \(A_e\) and retain the distinct pairs reached by
its proposal path. Let \(Q_e\) denote the number of these pairs. This is an ex
post, assignment-aware benchmark: it reveals only the opportunities used by
the full-graph DA computation, not a deployable recommendation score. Second,
we sample exactly \(Q_e\) edges uniformly without replacement from \(A_e\) and
run DA on that static graph. Third, we compute a maximum-cardinality matching
on the same sampled edges. We repeat the random draw 1,000 times within each
event, average first within event, and give equal weight to the 21 events.

The event, rather than the directed report or simulated graph draw, is the unit
of analysis. Each event supplies one market-level outcome after averaging over
the 1,000 seeded draws. Reported standard errors describe variation across the
21 markets; simulation draws reduce Monte Carlo error but do not create new
independent markets. Equal weighting prevents the largest events from
determining the estimand and matches the design's interpretation as a replay
across experimentally organized markets.

For a subgraph \(G\subseteq A_e\), let \(s_e(G)\) be the share of the short
side matched by DA and let \(a_e(G)\) be the share matched by a
maximum-cardinality matching. If \(G_{er}(Q_e)\) is random draw \(r\) with
exactly \(Q_e\) edges, the two primary estimands are
\begin{align}
 V^{\mathrm{path}}
 &=\frac{1}{|\mathcal E|}\sum_{e\in\mathcal E}
   \left[s_e(A_e)-\mathbb E_r s_e\bigl(G_{er}(Q_e)\bigr)\right],
 \label{eq:path-value-estimand}\\
 T^{\mathrm{stab}}
 &=\frac{1}{|\mathcal E|}\sum_{e\in\mathcal E}\mathbb E_r
   \left[a_e\bigl(G_{er}(Q_e)\bigr)
        -s_e\bigl(G_{er}(Q_e)\bigr)\right].
 \label{eq:stability-tax-estimand}
\end{align}
The first is the value of directing a fixed opportunity budget along the DA
path rather than scattering it randomly. The second is the same-graph cost of
stability. The identity
\begin{equation}
 \begin{aligned}
 V^{\mathrm{path}}
 &=\frac{1}{|\mathcal E|}\sum_{e\in\mathcal E}
   \left[a_e(A_e)-\mathbb E_r a_e\bigl(G_{er}(Q_e)\bigr)\right]
   +T^{\mathrm{stab}}\\
 &\quad-\frac{1}{|\mathcal E|}\sum_{e\in\mathcal E}
     \left[a_e(A_e)-s_e(A_e)\right].
 \end{aligned}
 \label{eq:empirical-loss-decomposition}
\end{equation}
expresses the path advantage as the feasibility loss from restricting
opportunities, plus the restricted-graph stability tax, minus the full-graph
stability tax. Each of these three components is nonnegative; their signed
sum need not be. This is match-count accounting, not cardinal welfare.

\begin{table}[!htbp]
\centering
\caption{Assignment-aware and random opportunity discovery in the speed-dating markets}
\label{tab:attention-application}
\small
\begin{tabular}{@{}lccccc@{}}
\hline\hline
Opportunity & Edges per & Full-info. & Random & Random & Stability \\
definition & proposer & DA path & DA & maximum & tax \\
\hline
Mutual consent & 1.14 (0.08) & 60.8 (2.4) & 51.0 (1.7) & 56.2 (1.9) & 5.3 (0.7) \\
All partners & 3.65 (0.29) & 100.0 (0.0) & 85.9 (1.6) & 95.3 (1.3) & 9.4 (0.8) \\
\hline\hline
\end{tabular}
\begin{minipage}{0.96\textwidth}
\footnotesize\emph{Notes:} The common budget counts executable pairs, not raw deliveries. The full-information DA path uses observed rankings and mutual acceptability to select proposals; it is an ex post benchmark. For each of 21 events, its proposal count sets the budget for both random rules, which receive identical uniformly sampled edges from the same primitive graph. Rates divide by the short side. Entries are equal-event means; standard errors in parentheses treat events as independent markets. All partners ranks every rated date above the outside option. Mutual consent retains pairs with two yes decisions. The replay holds rankings and decisions fixed; it is not a field treatment effect.
\end{minipage}
\end{table}

The preferred mutual-consent specification gives the economically stricter
comparison. The DA proposal path uses 1.14 executable opportunities per
proposer and matches 60.8 percent of the short side. Random DA with the same
number of opportunities matches 51.0 percent, whereas an unconstrained maximum
matching on those exact random edges reaches 56.2 percent. The 9.8-point path
advantage equals a 17.6-point feasibility loss plus a 5.3-point restricted-graph
stability tax, less a 13.1-point full-graph tax. It is positive in 18 of 21
events, negative in one, and zero in two. Holding each proposer's opportunity
count fixed as well reduces the average advantage to 3.4 points (Appendix
\ref{app:application-supplement}): part of the original gain comes from
allocating opportunities across proposers.

The broader all-partners specification isolates the assignment problem when
every recorded meeting is acceptable. The path matches the short side in every
event using 3.65 opportunities per proposer. Random DA at the same budget
matches 85.9 percent, although 95.3 percent is feasible on those same edges.
Of the 14.1-point path advantage, 9.4 points arise because stable clearing
requires a thicker graph and 4.7 points remain after allowing unconstrained
matching. Because every stable matching on a fixed graph has the same matched
set, changing the stable selection rule cannot remove the 9.4-point component.

\subsection{The empirical thickness frontier}

The equal-budget experiment evaluates one economically meaningful point: the
number of opportunities reached by full-graph DA. We next vary the revealed
share of each event's primitive graph. At every share, the stable and maximum
assignments receive exactly the same randomly sampled executable pairs. This
traces how quickly feasibility and stability convert graph thickness into
matches.

\begin{table}[!htbp]
\centering
\caption{Stable and feasible matching along the opportunity frontier}
\label{tab:attention-frontier}
\small
\begin{tabular}{@{}lrrrr@{}}
\hline\hline
Primitive graph & Edges per & Stable match & Maximum match & Stability \\
share & proposer & rate & rate & tax \\
\hline
\multicolumn{5}{@{}l}{\emph{A. All rated partners}} \\
10\% & 1.25 & 59.2 & 64.2 & 5.0 \\
30\% & 3.76 & 85.8 & 95.1 & 9.3 \\
50\% & 6.27 & 93.7 & 99.4 & 5.7 \\
80\% & 10.04 & 98.4 & 100.0 & 1.6 \\
100\% & 12.55 & 100.0 & 100.0 & 0.0 \\
\addlinespace[0.35em]
\multicolumn{5}{@{}l}{\emph{B. Mutual second-date consent}} \\
10\% & 0.21 & 17.1 & 17.3 & 0.2 \\
30\% & 0.66 & 37.8 & 40.1 & 2.3 \\
50\% & 1.10 & 47.1 & 52.4 & 5.3 \\
80\% & 1.75 & 56.6 & 66.2 & 9.6 \\
100\% & 2.19 & 60.8 & 73.8 & 13.1 \\
\hline\hline
\end{tabular}
\begin{minipage}{0.91\textwidth}
\footnotesize\emph{Notes:} Each row samples the stated share of executable edges uniformly without replacement within each event. Stable match is proposer deferred acceptance; Maximum match is a maximum-cardinality matching on exactly the same sampled graph. The stability tax is their percentage-point difference. Entries are unweighted means across 21 events and 1,000 seeded graph draws per non-full event cell. The complete eight-point frontier and wave-level results are in the replication package.
\end{minipage}
\end{table}

Panel A shows the nonmonotone same-graph stability gap predicted by the theory.
When all observed partners are acceptable, revealing 10 percent of the graph
creates 1.25 edges per proposer and leaves little room for either rule. At 30
percent, maximum matching reaches 95.1 percent while stable matching reaches
85.8 percent. The gap then closes as the graph becomes thick: at 80 percent the
rates are 100.0 and 98.4 percent, and both rules match the short side on the
complete graph. Stability is most costly between the extremes, where enough
opportunities exist to support many assignments but not enough to accommodate
the additional blocking constraints cheaply.

Panel B separates bilateral consent from exposure. The complete
mutual-consent graph contains only 2.19 executable edges per proposer on
average. Even after every consenting pair is included, maximum matching reaches
73.8 percent of the short side and stable matching reaches 60.8 percent. The
observed market can therefore lose matches at three distinct stages: some
meetings do not produce bilateral consent, a restricted policy may fail to
surface the right consenting pairs, and stability may prevent the clearing rule
from attaining the feasible matching size on the graph that remains.

The frontier preserves the realized preference heterogeneity of the 21 events
rather than imposing the independent random rankings of
Theorem~\ref{thm:attention-cost-stability}. The small event-level markets are
not used to estimate that theorem's asymptotic exponent. They show instead that
its central same-graph comparison is empirically consequential under observed
human rankings and bilateral decisions.

\subsection{Persistent application state in Chile}

Chile's national School Admission System supplies a complementary field
setting in which submitted rankings are carried from an initial centralized
allocation into a wait-list allocation \citep{mineduc2025sae}. The public 2024
regular-stage files contain 473,482 applicants and 1,638,432
applicant--school--course records, representing 1,613,351 applicant--school
pairs, including ranked applications, priorities,
school-specific lotteries and vacancies, initial assignments and responses,
and wait-list offers and responses. Families do not submit a new ranking
between the two reported allocations. An applicant can accept an initial
placement while keeping previously submitted, better-ranked options active.

\begin{table}[!htbp]
\centering
\caption{Consideration and retained reengagement in Chilean school choice}
\label{tab:chile-sae-attention}
\small
\begin{tabular}{@{}lrrrrr@{}}
\hline\hline
\multicolumn{6}{l}{\emph{Panel A: Submitted and local opportunity sets}} \\
\hline
Applicants & \multicolumn{5}{r}{473,482} \\
Mean voluntarily ranked schools & \multicolumn{5}{r}{3.04} \\
Share ranking at most three schools & \multicolumn{5}{r}{71.2\%} \\
Applicants in geocoded menu sample & \multicolumn{5}{r}{428,488} \\
Median nearby schools offering seats & \multicolumn{5}{r}{29} \\
Median share of local menu ranked & \multicolumn{5}{r}{6.8\%} \\
\hline
\multicolumn{6}{l}{\emph{Panel B: Outcomes from retained applications}} \\
Initial status & Retained better & Applicants & Favorable & Accepted & Accepted \\
 & ranked options & & offers & gains & gain rate \\
\hline
Accepted and waited & 0 & 599 & 0 & 0 & 0.0\% \\
Accepted and waited & 1 & 39,952 & 2,475 & 2,475 & 6.2\% \\
Accepted and waited & 2 & 19,518 & 1,775 & 1,775 & 9.1\% \\
Accepted and waited & 3-4 & 17,245 & 2,050 & 2,050 & 11.9\% \\
Accepted and waited & 5+ & 5,651 & 998 & 998 & 17.7\% \\
Initially unassigned & All listed & 34,444 & 3,958 & 2,204 & 6.4\% \\
\hline
Total retained & -- & 117,409 & 11,256 & 9,502 & 8.1\% \\
\hline\hline
\end{tabular}
\begin{minipage}{0.96\textwidth}
\footnotesize\emph{Notes:} Public 2024 regular-stage SAE data. Listed schools exclude options automatically added to preserve enrollment continuity. The local menu contains grade-compatible schools reporting initial vacancies within 5 km of the applicant's noisy geocode, after applying school gender restrictions. Listed schools is a mean; local menu and coverage are medians. Priority is the statutory priority classification. Preference positions rank course options; school counts deduplicate school identifiers. For accepted applicants, retained options are distinct other schools with a voluntarily ranked course above the actually assigned course; all other voluntarily ranked schools count when that course was an automatic continuity option. An accepted gain is a change of school to a strictly higher-ranked course option, a change from an automatic continuity course to a voluntary course at a different school, or a newly accepted ranked placement for an initially unassigned applicant. Local opportunity is not preference, and the wait-list transitions are administrative accounting rather than a causal estimate of LA-DA.
\end{minipage}
\end{table}

Panel A of Table~\ref{tab:chile-sae-attention} documents sparse submitted
opportunity sets. Excluding schools added automatically to preserve enrollment
continuity, applicants rank 3.04 schools on average and 71.2 percent rank at
most three. For applicants with usable noisy coordinates, we construct a local
menu from schools that report initial vacancies and satisfy grade and gender
restrictions. The median applicant has 29 such schools within 5 km but ranks
6.8 percent of that menu. A nearby unlisted school need not have been noticed,
acceptable, or preferred, so this comparison measures institutional
opportunity rather than latent welfare.

Panel B documents the economic use of persistent state. Among 82,965 applicants
who accepted an initial placement and elected to wait for a better submitted
choice, 7,298 later had an accepted move to a different school at a
higher-ranked course option, or from an automatic continuity course to a
voluntarily ranked course at another school. The accepted-gain
rate rises from 6.2 percent with one retained better-ranked option to 17.7
percent with five or more. Among 34,444 initially unassigned applicants who
remained in the process, 3,958 received a later offer and 2,204 accepted it. In
total, 9,502 of 117,409 retained applicants had an accepted improvement or new
placement implemented from their original ranking.

Two checks clarify the source and robustness of these transitions. The
Ministry's separate wait-list capacity file changes reported vacancies for
only 34 of 79,168 course records: 22 expansions add 22 vacancies and 12
contractions remove 150. These revisions are neither large nor plausibly
exogenous enough to define a natural experiment; responses and released seats
also change the assignment state. Separately, excluding 12,173 applicants in
administrative family blocks leaves 8,644 accepted gains among 105,236 retained
applicants, an 8.2 percent rate compared with 8.1 percent in the full sample.

These records do not identify the allocation that a mechanism without retained
rankings would have produced, and list length is endogenous. They establish a
narrower institutional fact: changing assignment conditions make earlier
applications consequential at scale, and the system implements those choices
from persistent authorization rather than requiring a new application round.
Appendix~\ref{app:external-evidence} gives the file construction and additional
evidence on the stages of the attention funnel.

\section{Exchange under Limited Attention}
\label{sec:la-ttc-extension}

Retaining valid decisions is also useful when the allocation objective is
exchange rather than bilateral stability. Consider a housing market with a
finite set \(N\) of owners, each endowed with object \(i\), and fixed strict
preferences over objects. A known graph \(G\subseteq N\times N\) specifies
eligible assignments and includes every self-edge. Keeping the endowment
is always possible. The reference rule is top trading cycles (TTC)
\citep{shapley1974cores,roth1977domination}.

\emph{Limited-attention top trading cycles} (LA-TTC) separates the information
needed to point from the authority needed to trade. For each remaining owner,
compare eligible remaining objects two at a time, retaining the preferred
one and recording the answer. Reuse valid recorded comparisons whenever that
choice recurs. The owner points to the original owner of the best object.
Identify the directed cycles, record their assignments, and remove their
owners and objects from the computation. Repeat until no owner remains.
These removals are computational: actual transfers wait until the necessary
evidence and permissions are established.

Assume that each required comparison eventually receives a truthful response
while its two alternatives are considered. Eligibility and preferences remain
fixed, evidence is retained, and every required exchange permission remains
valid through settlement. Such permission may come from a standing mandate
or from authorization of a specified conditional exchange; a preference
comparison alone does not supply it.

\begin{proposition}
\label{prop:la-ttc-extension}
Under these conditions, LA-TTC returns the reference TTC allocation on \(G\)
after at most \(|N|\) cycle-removal rounds, using no repeated successful
named comparison. The allocation is individually rational and in the strong
core of the restricted housing economy, hence Pareto efficient on \(G\).
\end{proposition}

Remembered comparisons give exactly the reference pointers, so the same
cycles are removed. Appendix~\ref{app:la-ttc-extension} proves the result,
states its formal-verification boundary, and reports a school-choice query
replay. The extension does not assert completion under an insufficient total
attention budget: binary comparisons require only two alternatives at once,
but eventual responses do not bound elapsed time. Nor does it permit an early
irreversible exchange merely because its
participants currently point to one another. Later information can change the
reference allocation. The common lesson is narrower: preserve valid evidence
and authority between episodes, and settle only after the rule's required
decisions have been established.

\section{Conclusion}

This paper studies the allocation cost of inattention in matching markets.
Classical matching theory begins with a graph of known opportunities. A
platform must create opportunities and maintain assignments as eligibility
changes. Valid authorization and current executability are different objects.
Remembering a proposal does not permit the mechanism to assign an ineligible pair.

Combining existing random-market limits gives a same-graph benchmark for what
execution cannot do: under homogeneous
thin attention, every stable mechanism leaves approximately
\(e^{-\sqrt d}\) unmatched where maximum matching leaves \(e^{-d}\). The
authorization lower bound shows what execution must remember: forgetting any
independently pivotal valid proposal forces either an allocation error or
another user query. Online LA-DA retains dormant applications and reengages
them when eligible, returning proposer DA on final consideration. Clearing
the full ledger is an optional extension with additional off-screen execution
rights. Assignment-aware recommendation changes the graph itself; with
exact next-best information and persistent executability, the random-market
benchmark gives a logarithmic saving in reached proposals.

The speed-dating evidence exposes the same margins. On mutually consenting
pairs, proposal-path discovery matches 60.8 percent of the short side at a
budget where random stable exposure matches 51.0 percent and random maximum
matching reaches 56.2 percent. Chilean school choice shows that the second
margin is active in an operating market: 9,502 applicants accept an assignment
gain from rankings retained beyond the initial allocation.

The general design principle is not to replace stable matching with
prediction. It is to direct scarce attention toward assignment-relevant
opportunities, retain valid authorizations, and clear only eligible pairs.
Discovery loss remains when a valuable edge is never created. Reengagement
avoids asking again for a still-valid action when that opportunity returns.

\label{sec:main-text-end}

\clearpage
\appendix
\section{Proof of the Static-Attention Impossibility Theorem}
\label{app:attention-cost-proof}

This appendix proves Theorem~\ref{thm:attention-cost-stability}. The stable
half reproduces the balanced case of \citet[Theorem 11]{arnosti2015short}
using exact threshold messages on the locally limiting marked Poisson tree.
Writing his matched share as \(x\) and setting his imbalance parameter to one
gives \(d=\log^2(1-x)/x\), equivalent to
\eqref{eq:stable-fixed-point}. The maximum-cardinality half specializes the bipartite formula
of \citet{bordenave2013matchings} to Poisson degrees.

\subsection{Reduction from mechanisms to matchings}

Fix a realized opportunity graph and preference profile. If \(\Phi_n\) is ex
post stable, its output is a stable matching on that realization. The
rural-hospitals, or lone-wolf, property implies that every stable matching of a
one-to-one market matches the same agents. Hence
\[
 |\Phi_n|=S_n(d)
\]
realization by realization, including when \(\Phi_n\) randomizes over stable
matchings. If \(\Psi_n\) is merely graph feasible, its output is a matching on
the same graph, so \(|\Psi_n|\leq A_n(d)\); maximum-cardinality matching attains
equality. It remains to derive the two random-graph limits and invert the stable
loss formula. This reduction is what makes the first part of the theorem an
impossibility for all stable clearing rules rather than a performance statement
about deferred acceptance alone.

\subsection{Stable threshold messages}

Fix a finite realization and a stable matching \(\mu\). For each acceptable
edge \((m,w)\), define
\[
 P_{mw}=1\quad\Longleftrightarrow\quad
 m\text{ weakly prefers }w\text{ to }\mu(m),
\]
and
\[
 H_{wm}=1\quad\Longleftrightarrow\quad
 w\text{ weakly prefers }m\text{ to }\mu(w).
\]
Equality here means that the edge is matched. Stability gives the exact
identities
\begin{align}
 P_{mw}
 &=\prod_{w'\succ_m w}(1-H_{w'm}),
 &
 H_{wm}
 &=\prod_{m'\succ_w m}(1-P_{m'w}).
 \label{eq:stable-message-identities}
\end{align}
For example, if \(m\) weakly prefers \(w\) to his assignment, no better
receiver \(w'\) can weakly prefer \(m\) to her assignment, or \((m,w')\)
would block. Conversely, if \(m\)'s assignment is better than \(w\), the
factor associated with that assignment is zero. The receiver identity is
symmetric.

An agent is unmatched exactly when every incident message from the other side
is zero. Hence
\begin{align}
 \mathbf 1\{m\text{ unmatched}\}
 &=\prod_{w\sim m}(1-H_{wm}),
 &
 \mathbf 1\{w\text{ unmatched}\}
 &=\prod_{m\sim w}(1-P_{mw}).
 \label{eq:unmatched-message-product}
\end{align}

The local weak limit of \(G_n(d)\), rooted at a uniformly selected endpoint,
is a two-type Poisson Galton--Watson tree with mean offspring \(d\). First
apply the recursion at finite depth with a constant boundary: incoming messages
use disjoint branches and are independent of the rank marks at the focal
endpoint. Follow a uniform directed edge. Conditional on its uniform rank \(u\),
the number of better incoming holder messages is Poisson with mean \(dhu\),
where \(h=\Pr(H=1)\). Averaging \(e^{-dhu}\) over \(u\in[0,1]\) gives
\[
 p=\phi(dh),
 \qquad
 h=\phi(dp),
 \qquad
 \phi(z)=\frac{1-e^{-z}}{z}.
 \label{eq:poisson-message-system}
\]
Define \(\phi(0)=1\) by continuity. The boundary squeeze below justifies
passing the finite-depth calculation to its limiting messages.

Multiplying the two equations by their denominators yields
\[
 dph=1-e^{-dh}=1-e^{-dp}.
\]
Strict monotonicity of \(1-e^{-z}\) implies \(p=h\). With \(t=dp=dh\),
\eqref{eq:poisson-message-system} becomes
\[
 t^2=d(1-e^{-t}).
\]
The map \(t\mapsto t^2/(1-e^{-t})\) is continuous and strictly increasing
from zero to infinity, so there is one positive solution \(t_d\).

\subsection{Boundary independence and finite-graph transfer}

Root the computation tree at a directed proposer message and expose an even
number \(2r\) of generations. Set every boundary proposer message to zero to
obtain \(P_r^-\), and to one to obtain \(P_r^+\); propagate
\eqref{eq:stable-message-identities} toward the root. A one-generation update
reverses coordinatewise order, so two updates preserve it. Consequently,
for every finite boundary condition and every realization of the marked tree,
\[
 P_r^-\le P_r\le P_r^+.
\]
The lower sequence increases with \(r\), the upper sequence decreases, and the
averaged two-generation map is
\[
 \Psi(p)=\phi\bigl(d\phi(dp)\bigr).
\]
Their limiting marginal probabilities are fixed points of \(\Psi\). If
\(p=\Psi(p)\) and \(h=\phi(dp)\), then \((p,h)\) solves
\eqref{eq:poisson-message-system}; the uniqueness argument above therefore
makes the lower and upper marginal limits equal. Since the messages are binary
and coupled with \(P_r^-\le P_r^+\),
\begin{equation}
 \Pr(P_r^-\ne P_r^+)
 =\mathbb E[P_r^+-P_r^-]\longrightarrow0.
 \label{eq:boundary-disagreement}
\end{equation}
The same holds for receiver messages and, by a union bound over the finite-
mean root degree, for the root unmatched indicator.

For fixed \(r\), the marked radius-\(2r\) neighborhood of a uniformly sampled
finite-market agent converges to this tree and is acyclic with probability
tending to one. Whatever messages the rest of the finite graph induces at the
boundary, the true stable messages are squeezed between the two extremal
solutions. First taking \(n\to\infty\) at fixed \(r\), and then
\(r\to\infty\), Equation~\eqref{eq:boundary-disagreement} therefore implies
that the unmatched probability converges to its tree value.
At the Poisson root, the degree \(D\) is Poisson with mean \(d\), and its
incoming holder messages are conditionally independent Bernoulli \(h\). Hence
\[
 \Pr\{\text{root unmatched}\}
 =\mathbb E[(1-h)^D]
 =e^{-dh}=e^{-t_d}.
\]

For concentration, sample two distinct roots. Their fixed-radius
neighborhoods are asymptotically disjoint and converge jointly to independent
Poisson trees. The same boundary squeeze makes the covariance of their
unmatched indicators vanish. Exchangeability then gives convergence in
probability of the empirical unmatched share. The rural-hospitals property
and the reduction above extend the conclusion from proposer DA to every ex
post stable allocation rule
\citep{roth2008deferred}.

\subsection{Maximum-cardinality benchmark}

For Poisson degree generating function
\(\varphi(z)=e^{d(z-1)}\), Theorem 3 of
\citet{bordenave2013matchings} gives
\[
 \frac{A_n(d)}n\xrightarrow{p}1-\max_{0\le s\le1}F_d(s),
\]
where \(F_d\) is defined in \eqref{eq:max-fixed-point}; the normalization is
the share of either market side covered by a maximum matching. This is the
same graph sequence used for the stable calculation, so the comparison does
not change opportunity supply.

\subsection{Large-thickness asymptotics}

The stable fixed point implies \(t_d\to\infty\) and
\[
 \frac{t_d}{\sqrt d}=\sqrt{1-e^{-t_d}}\longrightarrow1.
\]
Therefore \(e^{-t_d}=e^{-\sqrt d(1+o(1))}\).

For maximum matching, put \(q=e^{-ds}\). Then
\[
 F_d(s)=f_d(q)=e^{-dq}-1+q(1-\log q),
 \qquad e^{-d}\le q\le1.
\]
Both endpoint candidates are \(e^{-d}(1+o(1))\). At an interior extremum,
\[
 -\log q=d e^{-dq}.
\]
Writing \(z=dq\), a local maximum also satisfies
\(z\log(d/z)\le1\). These two restrictions force either
\[
 z=d e^{-d}(1+o(1))
 \quad\text{or}\quad
 z=d-d^2e^{-d}(1+o(1)).
\]
Indeed, if \(z\le d/2\), the second-order condition gives
\(z\le C=1/\log 2\). Stationarity then implies
\(z\le d\exp(-d e^{-C})\), so \(dz\to0\). Substituting back gives
\(\log(d/z)=d+o(1)\), which yields the first branch with relative error
\(o(1)\). If \(z>d/2\), write \(z=d-r\). The second-order condition gives
\(r\le2\). Expanding the exact stationary equation yields
\(r=d^2e^{-d+r}(1+O(1/d))\). First \(r=O(d^2e^{-d})=o(1)\), and then
\(r=d^2e^{-d}(1+o(1))\), proving the second branch.
At a stationary point,
\(f_d(z/d)=e^{-z}(1+z)-1+z/d\).
On the small branch this is \(z/d-z^2/2+O(z^3)\); on the large branch,
expanding the original expression gives
\(f_d(1-r/d)=e^{-d+r}-r^2/(2d^2)+O(r^3/d^3)\).
Both equal \(e^{-d}(1+o(1))\). Any remaining stationary
points fail the local-maximum condition, while the two endpoints also have
value \(e^{-d}(1+o(1))\). Hence
\[
 \max_s F_d(s)=e^{-d}(1+o(1)).
\]

Finally set \(L=\log(1/\varepsilon)\). The stable identity
\(\varepsilon=e^{-t}\) gives
\[
 d=\frac{t^2}{1-e^{-t}}=\frac{L^2}{1-\varepsilon}.
\]
The maximum-matching asymptotic gives \(d=L(1+o(1))\). Since
\(\varepsilon_{\mathrm{stab}}(d)\) is strictly decreasing, any
\(d<L^2/(1-\varepsilon)\) has
\(1-\varepsilon_{\mathrm{stab}}(d)<1-\varepsilon\). Convergence in
probability implies
\(\Pr\{S_n(d)/n\geq1-\varepsilon\}\to0\). Whenever \(\Psi_n\)'s output is
stable, matched-set invariance gives \(|\Psi_n|=S_n(d)\). Hence the event that
\(\Psi_n\) returns a stable matching and matches at least a
\(1-\varepsilon\) share is contained in the preceding vanishing-probability
event. This completes the proof of Theorem~\ref{thm:attention-cost-stability}.

\section{LA-DA Specification and Verification}
\label{app:formal-verification}

This appendix gives the complete finite controller, proves Theorem
\ref{thm:main} in paper notation, and then records the machine-checked boundary.
It distinguishes that online result from Corollary~\ref{cor:offscreen-settlement}.
The large-market limit in Theorem~\ref{thm:attention-cost-stability} is proved
mathematically in Appendix~\ref{app:attention-cost-proof} rather than in Lean.

\subsection{Operational specification and allocation proof}

This subsection makes the mechanism in Section
\ref{sec:sequential-implementation} self-contained. Fix an admissible input
\(\theta=(\succ,k,C^0,\rho)\). For \(B\subseteq W\), let
\(U_m(B,w)\) add \(w\) when \(|B|<k_m\); when \(|B|=k_m\), it keeps the
\(k_m\) best elements of \(B\cup\{w\}\) according to \(m\)'s strict order.
At a holding boundary, the state is \((\widetilde\mu,C,K,U)\): the tentative
matching, current on-screen consideration, retained proposal correspondence
\(K=(K_w)_{w\in W}\), and receivers available in the current pass. We also
identify \(K\) with its edge graph \(\{(m,w):m\in K_w\}\). The tentative
matching guides discovery; \(C\) constrains execution and \(K\) preserves authorization.
LA-DA is the following finite controller.

\begin{enumerate}
  \item Initialize \(\widetilde\mu\) as unmatched, set \(C=C^0\), set every \(K_w\)
  empty, and set \(U=W\).

  \item Process the restriction of \(\rho\) to \(U\). When \(w\) is reached,
  replace each \(C_m\) by \(U_m(C_m,w)\).

  \item After that recommendation, each proposer records at most one new
  application: his best \(w'\in C_m\) such that
  \(w'\succ_m\widetilde\mu(m)\), \(w'\succ_m\emptyset\), and
  \(m\notin K_{w'}\).
  All proposers choose from the same pre-batch state, and each selected pair is
  added permanently to the proposal record while its authorization remains
  valid.

  \item Recompute \(\widetilde\mu\) by proposer DA on the active retained graph
  \(E(C,K)\). Remove from \(C_m\) every unmatched active proposal that \(m\)
  still prefers to his assignment and whose receiver prefers her current
  holder. Do not delete its record from \(K\). A later recommendation can
  therefore reactivate that proposal without another application.

  \item After the pass, let \(U'\) be all currently unmatched receivers. If a
  profitable considered pair remains unrecorded, the pass recorded a new pair,
  and \(U'\ne\varnothing\), set \(U=U'\) and start another pass. Otherwise
  freeze recommendation and enter terminal closure.

  \item Terminal closure recomputes holders, records simultaneously each
  proposer's best profitable considered unrecorded pair, and repeats the holder
  and rejection update until no such pair remains. Online LA-DA returns
  \(\mu^T=\widetilde\mu^T\), together with \(C^T,K^T\).
\end{enumerate}

The terminal state satisfies \(|C_m^T|\le k_m\), holders equal proposer DA on
\(E_T=E(C^T,K^T)\), and every profitable pair still considered has been
recorded. These are the certificates for online clearing, not a license to
execute the dormant records. Under Corollary~\ref{cor:offscreen-settlement},
an optional additional step instead returns
\(\widehat\mu^T=\operatorname{DA}_M(G_T)\). It uses all valid ledger pairs
and requires the additional comparisons and off-screen execution rights.

The baseline theorem fixes preferences and mutual acceptability during a run.
A revoked or expired pair must be removed from the executable market before a
settlement; arbitrary within-run revocation is not part of \(\theta\).
Temporary disappearance from current consideration is part of the controller
and does not erase a still-valid proposal record.

\begin{proof}[Proof of Theorem~\ref{thm:main}]
Each pair enters \(K\) at most once. Another serial pass begins only after the
preceding pass records a new pair, and each terminal-closure iteration also
records a new pair. There are at most \(|M||W|\) records, while every pass and
holder closure is finite. The controller therefore terminates.

Capacity updates keep at most \(k_m\) options and rejection only removes
options. At terminal closure, normalized holders give
\(\mu^T=\operatorname{DA}_M(E_T)\). Put \(F_T=A\cap C^T\).
Every edge in \(F_T\setminus E_T\) is unrecorded and, by exhaustion, is not
preferred by its proposer to \(\mu^T\). Thus \(\mu^T\) is stable on \(F_T\).

Let \(\nu=\operatorname{DA}_M(F_T)\). Proposer optimality implies that each
proposer weakly prefers \(\nu\) to \(\mu^T\). Every edge of \(\nu\) must
lie in \(E_T\): an unchanged partner is already there, and a strictly better
omitted partner would contradict exhaustion. Hence \(\nu\) is feasible and
stable on \(E_T\). Proposer optimality of \(\mu^T\) on \(E_T\) gives the
reverse comparison; strictness implies \(\nu=\mu^T\).
Finally, any feasible Pareto improvement would contain a changed pair whose
endpoints both strictly gain, contradicting stability. This proves
\eqref{eq:la-da-online-allocation} and its welfare conclusions.
\end{proof}

Corollary~\ref{cor:offscreen-settlement} follows by applying ordinary proposer
DA to \(G_T\) under its extra execution rights. It is a different output,
not an additional step needed to establish Theorem~\ref{thm:main}.

\subsection{Execution rights and the attention frontier}
\label{app:execution-rights-frontier}

Suppose off-screen settlement is permitted and recommendations accumulate a
retained graph distributed as \(G_n(d)\) in Section~\ref{sec:attention-cost}.
During discovery, allow profiles to disappear temporarily while retaining
valid records, and then settle the ledger. Corollary
\ref{cor:offscreen-settlement} and Theorem~\ref{thm:attention-cost-stability}
imply
\[
 1-\frac{|\widehat\mu^T|}{n}
 \xrightarrow{p}\varepsilon_{\mathrm{stab}}(d).
\]
Thus the accumulated-graph frontier applies under persistent execution rights.
With binding current eligibility, Theorem~\ref{thm:main} instead concerns
\(F_T\); identifying its distribution requires a separate argument.

The results distinguish loss from the generated opportunity graph from loss
due to forgetting or incorrect dynamic clearing. Recommendation changes the
first margin; online LA-DA addresses the second on the eligible graph. Retained
authorization avoids paying again for a prior action when eligibility returns.
It does not make dormant opportunities currently feasible.

\subsection{Minimal operational state}

Two holding-boundary histories are \emph{continuation equivalent} when every
common admissible future event sequence produces the same terminal allocation.
Their equivalence classes form the canonical minimal implementation state. For
the concrete controller, define
\[
 Z(h)=(C(h),K(h),U(h),r(h)),
\]
where \(U\) is the set of receiver columns still available in the current pass
and \(r\) is the round counter. The cached matching carries no independent
history at a holding-closed boundary because it is reconstructed as
\(\operatorname{DA}_M(E(C,K))\). Under the transition rules above, a common
future applied to equal \(Z\)-states produces equal subsequent states by
induction. Thus \(Z\) is sufficient for execution.

Theorem~\ref{thm:authorization-memory} identifies the coordinate that cannot in
general be compressed away. In its \(k\)-component construction, the
authorization vector belongs to \(\{0,1\}^k\), and exact query-neutral execution
recovers that vector from memory. The encoder is therefore injective and its
range has cardinality at least \(2^k\).

The Lean file
\path{MatchingLimitedConsiderationFormal/ImplementationState.lean} checks
each step of this argument. It constructs the continuation-equivalence
quotient, proves that every sufficient representation refines it, establishes
deterministic continuation sufficiency, and proves the \(2^k\)-state lower
bound. The replication documentation maps these statements to their exported
Lean declarations.

\subsection{Lean project and build}

The formal project is in \texttt{formal/}. It pins Lean 4.33.1 and imports
EconCSLib's one-to-one matching definitions. The complete verification command
is
\begin{verbatim}
cd formal
lake build
\end{verbatim}
The project files contain no \texttt{sorry}, \texttt{admit}, or local axiom
placeholders on the exported proof path.

\subsection{Online controller and optional settlement in Lean}

The state contains the tentative assignment, current consideration,
retained proposals, recommendation availability, and program location.
The formalization preserves its original internal identifiers for backward
compatibility. \texttt{Interleaved.lean} defines the total online controller
\begin{verbatim}
interleavedDA_MSRR
\end{verbatim}
using explicit well-founded recursion. The two recursive loops decrease the
finite retained-pair deficit. Lean checks terminal exhaustion, certificate
generation, capacity compliance, and the active-screen controller theorem:
\begin{verbatim}
interleavedDA_MSRR_terminalExhaustion
interleavedDA_MSRR_generatesCertificate
interleaved_da_msrr_consideration_correct
\end{verbatim}

\texttt{AllocationCharacterization.lean} identifies that same output with
proposer DA on final consideration:
\begin{verbatim}
interleavedDA_MSRR_assignment_eq_finalConsiderationDA
\end{verbatim}
Together these declarations support Theorem~\ref{thm:main}. The separate
persistent-authorization module defines the cumulative \texttt{candidateSet}
graph and the optional settlement in Corollary~\ref{cor:offscreen-settlement}:
\begin{verbatim}
pandaTerminalGraph
panda
panda_eq_proposerDA
panda_isProposerOptimalStable
panda_isParetoEfficient
panda_pathIndependent
\end{verbatim}
These historical \texttt{panda} identifiers denote the off-screen extension,
not online LA-DA's return value. Its terminal graph filters the retained ledger only by mutual acceptability;
it does not depend on final on-screen consideration. Thus
\texttt{panda\_pathIndependent} checks the conditional Corollary
\ref{cor:terminal-graph-separation} even for runs with different capacities,
initial consideration, and recommendation histories.

Python exposes the online output through
\texttt{run\_interleaved\_da\_msrr}; \texttt{run\_panda} retains the
two-stage extension for reproducibility. Neither public implementation is
silently changed by the distinction in the manuscript. Lean checks the
functional specifications; these real-valued definitions are
\texttt{noncomputable}, while the Python versions provide executable replays.

The Lean definitions use real-valued representatives of strict ordinal
preferences. The outside-option value is a normalization; no interpersonal
or cardinal comparison is used in the allocation theorem.

\subsection{Optimal discovery and query-neutral execution}

\texttt{ExpectedBellman.lean} verifies the finite-horizon expected-cost
Bellman recursion and the optimality of the policy it generates.
\texttt{PotentialCoreBellman.lean} verifies the structural reductions used in
Theorem~\ref{thm:potential-core-attention-optimality}: an everywhere-pivotal
query can be moved to the front without increasing realized cost, independent
target blocks satisfy an exact expected-cost direct sum, and a compatible
strict-mutual-top edge is pivotal for proposer DA. The principal declarations
are
\begin{verbatim}
frontload_pointwisePivotal_is_expectedCostOptimal
sequentialProduct_is_expectedCostOptimal
frontload_mutualTop_DA_query_is_expectedCostOptimal
\end{verbatim}

\texttt{AttentionExecutionSeparation.lean} verifies
Theorem~\ref{thm:attention-execution-separation}. It maps a deterministic
settlement rule over the leaves of an arbitrary binary discovery tree and
proves that the realized query sequence, every path cost, and expected cost
are unchanged. It then establishes both the information lower bound and
attainment by certified execution:
\begin{verbatim}
optimalTargetTree_le_composedController
certifiedExecution_of_optimalDiscovery
pandaSettlement_eq_terminalGraphDA
certifiedPANDA_of_optimalDiscovery
\end{verbatim}

For finite audit markets, the Python optimizer exports its complete query DAG
to a generated Lean certificate. Lean replays every hidden compatibility
realization, independently computes proposer DA, and compares the
certificate's expected cost with unrestricted Bellman. Regenerate and check
that certificate with
\begin{verbatim}
empirical/.venv/bin/python \
  formal/scripts/generate_potential_core_certificate.py --check
\end{verbatim}
The large finite replays use \texttt{native\_decide} and are supplemental
compiler-backed audits. They are not dependencies of the kernel-checked,
axiom-audited paper theorems.

\subsection{Endogenous recommendation}

\texttt{EndogenousRecommendation.lean} defines a structurally recursive
serial blocker-or-certificate controller. It clears the remembered graph,
selects a proposer with an omitted primitive block, adds that proposer's best
omitted willing challenger, and clears again. The missing-edge count strictly
decreases. Lean checks:
\begin{verbatim}
endogenousRecommendationStep_decreases
endogenousRecommendationTerminalGraph_stops
endogenousRecommendationTerminalGraph_covers_blocks
endogenousRecommendation_correct
endogenousRecommendation_full_market_correct
\end{verbatim}
The exported theorem proves that the terminal matching is stable on the full
primitive graph and Pareto efficient. This is the serial specialization of
Theorem~\ref{thm:endogenous-recommendation}; the paper statement allows the
same finite additions to be batched.

The formalization assumes that the selected favorite is correct and that an
added edge remains executable, not merely remembered. Preference certification and
consent are economic inputs to the controller, not consequences of Lean's
type system. Section~\ref{sec:top-choice-identification} supplies one
mathematical route for the first input.

For Theorem~\ref{thm:adaptive-attention-gap}, Lean checks the optional ledger-clearing
back end. The next-best trace lemma is paper-proved, and the random-market
proposal and admissibility rates are cited probabilistic results rather than
new machine-checked claims.

The budgeted planning result in
Theorem~\ref{thm:budgeted-recommendation-hardness} is a mathematical Knapsack
reduction rather than a Lean theorem. Proposition
\ref{prop:complete-order-correction} records the implementation boundary that
matters for the paper: a precedence reduction for the online
active-screen assignment does not establish hardness for optional
cumulative-ledger settlement. The
legacy fixed-order regression files remain in the repository as an audit of
the active-screen object; they are not evidence for the paper-facing
complexity theorem.

\subsection{Executable audits}

The large-market and finite-query calculations can be reproduced with
\begin{verbatim}
python3 theory/computation/delegated_inspection/audit_thickness_cost.py
python3 theory/computation/delegated_inspection/audit_query_complexity.py
python3 theory/computation/delegated_inspection/audit_hodge_preference_cycles.py
python3 theory/computation/delegated_inspection/audit_potential_core_bellman.py
\end{verbatim}
The first script compares finite random markets with the fixed-point formulas.
The second audits finite target-identification complexity, the third reports
the preference-cycle decomposition, and the fourth compares the potential-core
policy with unrestricted Bellman on enumerated finite markets. These
executable audits supplement, but do not replace, the mathematical proofs.

\section{Identifying the Next Proposal}
\label{sec:top-choice-identification}

Under the low-dimensional preference structure developed below, limited
consideration need not prevent a platform from discovering an agent's
most-preferred feasible candidates. The platform does not need to reconstruct
every profile the agent considered, recover a complete preference order, or
point identify every taste coefficient.  It needs a smaller object: whenever a
proposal is required, it must identify the proposer's best currently eligible
option that has not already rejected him.  We call this \emph{choice-functional
identification}.

\subsection{Low-dimensional tastes and approval reports}

Every receiver \(w\in\W\) has a profile representation
\(e_w\in\mathbb R^{d_x}\), where \(d_x\) is fixed and the representation is
constructed before the reports studied below are used. The
baseline proposer utility is the pure-characteristics index
\begin{equation}
  u_m(w)=\beta_m'e_w,
  \qquad
  \beta_m\overset{\mathrm{iid}}{\sim}N(\bar\beta,\Sigma).
  \label{eq:pure-characteristics-preference}
\end{equation}
Observed proposer characteristics can shift the conditional mean without
changing the argument.  Writing \(\beta_m=\bar\beta+\nu_m\) separates a
vertical component, \(\bar\beta'e_w\), from horizontal compatibility,
\(\nu_m'e_w\).  In the Gaussian baseline,
\begin{equation}
 \Pr\{u_m(w)>u_m(v)\}
 =\Phi\!\left(
   \frac{\bar\beta'(e_w-e_v)}
        {\sqrt{(e_w-e_v)'\Sigma(e_w-e_v)}}
 \right).
 \label{eq:vertical-horizontal-choice}
\end{equation}
when the denominator is positive.
We maintain distinct indices on each finite candidate set.  This holds almost
surely when every profile contrast has positive variance; fixed deterministic
tie breaking covers null events.
This is the Gaussian special case of the elliptical random-coefficient
geometry used by \citet{he2025random} to distinguish vertical and horizontal
differentiation from binary comparisons.  Equation
\eqref{eq:pure-characteristics-preference} is substantive: extrapolation fails
if unobserved pair-specific shocks dominate the characteristics index.
The deterministic identification result below uses the linear index and the
report-consistent taste region; normality supplies a convenient population and
Bayesian specialization but is not required for the geometric argument.

Suppose a processed profile produces one of three ordered reports.  For
thresholds \(\tau_m^L<\tau_m^S\), an explicit dislike, ordinary like, and
priority like respectively imply
\begin{equation}
 \beta_m'e_w<\tau_m^L,
 \qquad
 \tau_m^L\leq\beta_m'e_w<\tau_m^S,
 \qquad
 \beta_m'e_w\geq\tau_m^S.
 \label{eq:approval-inequalities}
\end{equation}
A direct comparison \(w\succ_m v\) adds
\(\beta_m'(e_w-e_v)>0\).  The distinction between approval and list traversal
in \citet{manzini2024approval} motivates treating these reports as coarser than
a submitted ranking.  If priority tokens are scarce, their clean interpretation
is relative---the selected profile is ranked above profiles that could have
received the same token---rather than a fixed cardinal threshold.

Let \(H_m\) collect only reports tied to verified processing, and let
\(\mathcal B_m(H_m)\) be the set of taste coefficients satisfying their
inequalities.  A missing response imposes no restriction unless processing is
known: silence may be nonattention rather than dislike.  Unknown thresholds
can be included in an augmented parameter vector.  A Bayesian implementation
instead uses the Gaussian prior in
\eqref{eq:pure-characteristics-preference} and replaces
\(\mathcal B_m(H_m)\) with a posterior uncertainty region.

\subsection{Identifying a choice without identifying preferences}

For a nonempty set \(R\subseteq\W\) of currently eligible and unrejected
receivers, define the potential-top set
\begin{equation}
 \mathcal T_m(H_m,R)
 =\bigcup_{\beta\in\mathcal B_m(H_m)}
   \arg\max_{w\in R}\ \beta'e_w.
 \label{eq:potential-top-set}
\end{equation}
The set contains every receiver who is optimal under some taste vector still
consistent with the attended reports.  The target is point identified when
\(\mathcal T_m(H_m,R)\) is a singleton, even if
\(\mathcal B_m(H_m)\) and the remaining preference comparisons are not.
For each \(w\in R\), define its normal cone in the profile geometry by
\begin{equation}
 N_R(w)=\left\{\beta:
   \beta'(e_w-e_v)\geq 0\ \text{for every }v\in R\right\}.
 \label{eq:profile-normal-cone}
\end{equation}
The cones are the regions of taste space that select each possible top
candidate.  Locating a taste region relative to these cones is a finite
classification problem; recovering the numerical coefficient is unnecessary.

\begin{proposition}
\label{prop:top-choice-identification}
For every report history and candidate set,
\begin{equation}
 \mathcal T_m(H_m,R)
 =\left\{w\in R:
   \mathcal B_m(H_m)\cap N_R(w)\neq\varnothing\right\}.
 \label{eq:potential-top-normal-cones}
\end{equation}
Suppose the true \(\beta_m\) belongs to \(\mathcal B_m(H_m)\).  Then the true
maximizer of \(u_m\) on \(R\) belongs to \(\mathcal T_m(H_m,R)\).  Hence a
singleton potential-top set identifies the true top option.

If \(H_m^r\subseteq H_m^{r+1}\) are cumulative reports from nested
recommendation histories, then
\[
 \mathcal B_m(H_m^{r+1})\subseteq\mathcal B_m(H_m^r),
 \qquad
 \mathcal T_m(H_m^{r+1},R)\subseteq
 \mathcal T_m(H_m^r,R).
\]
For nonempty compact \(\mathcal B_m(H_m)\), a candidate \(w^*\in R\) is the unique
certified top option whenever
\begin{equation}
 \inf_{\beta\in\mathcal B_m(H_m)}
   \beta'(e_{w^*}-e_w)>0
 \quad\text{for every }w\in R\setminus\{w^*\}.
 \label{eq:robust-top-margin}
\end{equation}
\end{proposition}

\begin{proof}
By \eqref{eq:profile-normal-cone}, \(w\) maximizes the linear index at
\(\beta\) if and only if \(\beta\in N_R(w)\), which proves
\eqref{eq:potential-top-normal-cones}.  The true coefficient is one of the
coefficients over which the union in \eqref{eq:potential-top-set} is taken,
proving coverage of the true maximizer.  Adding reports adds inequalities, so
the feasible coefficient region contracts; taking the union of maximizers over
a smaller region cannot add an option.  Finally,
\eqref{eq:robust-top-margin} makes \(w^*\) strictly better than every other
candidate for every coefficient in the region.
\end{proof}

This target is minimal in a literal sense.  On a fixed remaining set and under
the maintained strictness or fixed tie breaking, two taste vectors are
equivalent for the next proposal exactly when they select the same candidate.
Any exact recommendation statistic must distinguish two vectors that select
different candidates, but it need not distinguish vectors inside the same
normal-cone cell.  The choice functional is therefore the coarsest information
about tastes sufficient for that proposal.

The result identifies neither the realized historical consideration set nor a
complete ranking.  Nesting is useful because remembered reports monotonically
remove possible maximizers.  Nesting alone, however, does not separate
attention from taste.  Valid estimation still requires recorded processing,
exogenous menu or order variation, known exploration propensities, or another
maintained attention model
\citep{manzini2014stochastic,abaluck2021consumers,aguiar2023random}.

\subsection{Sequential elicitation is sufficient for deferred acceptance}

Full preference lists are also unnecessary for clearing.  Fix a finite
eligible graph and receiver rankings.  A \emph{next-best oracle} takes a
proposer and his set of acceptable receivers who have not rejected him and
returns his favorite member of that set.  Lazy deferred acceptance calls this
oracle only when an unmatched proposer needs to act.

\begin{proposition}
\label{prop:next-best-sufficiency}
Fix a finite marriage market with strict preferences and a proposer activation
schedule.  If every next-best oracle response is correct, lazy proposer
deferred acceptance generates the same proposal trace and terminal matching as
full-list proposer deferred acceptance under that schedule.  Consequently,
preferences among alternatives never reached by the proposal process need not
be recovered.
\end{proposition}

\begin{proof}
Initially both procedures give every proposer the same unrejected set.  Suppose
their traces agree before a proposal event.  They then call on the same
proposer, whose oracle returns the first remaining receiver on his true list;
the proposal and the receiver's holding decision therefore agree.  The two
procedures remove the same rejected edge and reach the same next state.
Induction proves equality through termination.
\end{proof}

One identified top option is sufficient until it rejects the proposer.  The
platform then removes it and identifies the top of the remaining set.  Thus the
necessary information is an adaptively queried preference prefix, not a full
ranking.

For statistical certification, let
\(\widehat{\mathcal B}_m(H_m)\) be an anytime-valid confidence region satisfying
\begin{equation}
 \Pr\!\left\{
   \beta_m\in\widehat{\mathcal B}_m(H_m)
   \text{ for every proposer and every history reached by the policy}
 \right\}\geq1-\alpha.
 \label{eq:anytime-taste-coverage}
\end{equation}
On this event, every singleton certification from
\eqref{eq:potential-top-set} is correct.  Proposition
\ref{prop:next-best-sufficiency} therefore transfers the same probability to
equality with full-information proposer DA whenever the policy obtains a
singleton before each required proposal.  If the set is not a singleton, the
platform must collect another informative report, request a direct comparison,
or return a set-valued recommendation.  A point prediction alone does not
retain the guarantee.

\subsection{Learn, search, certify, clear}

The platform can now separate four tasks.  It \emph{learns} a taste region from
attended approvals; \emph{searches} the unobserved market using profile
representations; \emph{certifies} only comparisons shared by all admissible
tastes; and \emph{clears} mutually acceptable, authorized opportunities.
Exploration and execution use different graphs.  A profile may be shown to
learn; an edge enters the executable graph only after eligibility and bilateral
authorization are established.

LA-DA is the rolling implementation of the final task. At each proposal
event it needs only the best active, profitable, unrecorded proposal.  If the
rankings used for these events are directly reported or certified by
\eqref{eq:robust-top-margin}, Theorem~\ref{thm:main} applies to the generated
executable graph.  Application memory preserves an authorized edge while it is
dormant; preference memory avoids repeated elicitation.

\section{Budgeted Discovery and Global Plans}
\label{app:budgeted-discovery}

The recommendation problem becomes combinatorial when the platform must stop
before processing the full market. This appendix gives the narrow complexity
result used in the text and clarifies why complete-order complexity is not a
property of cumulative-ledger settlement. Throughout this appendix, settlement
has the additional execution rights of Corollary~\ref{cor:offscreen-settlement}.

Let a global discovery plan \(\rho\) be a no-duplicate list of receiver
columns, not necessarily containing every receiver. Processing column \(w\)
costs the platform a nonnegative integer \(c_w\). Discovery stops after
\(\rho\), LA-DA retains every valid proposal generated along that path, and
the platform settles the resulting cumulative graph by proposer DA. Given
nonnegative integer match values \(z_{mw}\), define
\[
 V(\rho)=\sum_{m:\mu^\rho(m)\ne\emptyset}z_{m,\mu^\rho(m)},
 \qquad
 C(\rho)=\sum_{w\in\rho}c_w,
\]
where \(\mu^\rho\) is the retained-ledger settlement. The decision problem asks
whether some plan satisfies \(C(\rho)\le B\) and \(V(\rho)\ge R\).

\begin{theorem}
\label{thm:budgeted-recommendation-hardness}
The budgeted recommendation decision problem is NP-complete. Its restriction
to empty initial consideration, unit display capacity, and isolated bilateral
markets is weakly NP-complete: in this restriction, proposer \(m_j\) and
receiver \(w_j\) are one another's only acceptable partners.
\end{theorem}

\begin{proof}
A no-duplicate plan is a polynomial-size certificate. The finite controller
records each proposal at most once, and proposer DA on its retained graph is
polynomial, so the problem belongs to NP.

Reduce from zero-one Knapsack. For item \(j\) with integer weight \(d_j\) and
value \(q_j\), create an isolated mutually acceptable pair
\((m_j,w_j)\), set \(c_{w_j}=d_j\), and set
\(z_{m_jw_j}=q_j\). Recommending \(w_j\) records its unique valid application;
not recommending it leaves that component empty. Hence a plan processing
columns \(S\) costs \(\sum_{j\in S}d_j\) and settles with value
\(\sum_{j\in S}q_j\). A feasible plan exists exactly when the Knapsack
instance has a feasible subset. This proves NP-hardness of the general
problem; membership in NP gives NP-completeness. In the isolated-pair
restriction, the converse translation into Knapsack is immediate. Its
pseudo-polynomial algorithm establishes the stated weak classification for
that restriction, not for the general problem.
\end{proof}

The reduction deliberately contains no rejection chain or preference cycle.
It establishes the computational cost of allocating a heterogeneous discovery
budget, not hardness created by LA-DA. It also gives the tractable boundary.

\begin{corollary}
\label{cor:isolated-budget-planning}
In the isolated-pair market, the optimal plan is exactly a zero-one Knapsack
solution. It therefore admits the standard pseudo-polynomial dynamic program
and fully polynomial approximation scheme. If all column costs are one and at
most \(B\) columns may be processed, an optimum selects the \(B\) largest
positive match values.
\end{corollary}

The result does not imply that choosing a permutation of a mandatory complete
receiver list is hard. In fact, the sparse construction previously considered
for that purpose gives the opposite conclusion once settlement uses the
cumulative ledger.

\begin{proposition}
\label{prop:complete-order-correction}
Create one home pair \((m_j,w_j)\) for each \(j\), let \(w_j\) accept only
\(m_j\), and allow proposer \(m_j\) to rank any number of other receivers above
\(w_j\). Start with empty consideration and an empty ledger, give every
proposer display capacity one, and process every receiver column. Under
cumulative-ledger settlement, every home proposal is retained and every home
pair is matched, regardless of the receiver order.
\end{proposition}

\begin{proof}
Every nonhome receiver rejects \(m_j\), because she accepts only her own
proposer. Such a rejection removes that receiver from \(m_j\)'s active display.
When \(w_j\) is processed, proposer \(m_j\) therefore either already has his
home authorization or is unmatched and records it. Retention keeps
\((m_j,w_j)\) in the terminal ledger even if a later unacceptable profile
temporarily replaces it on screen. The valid terminal graph consequently
contains every home edge and no valid nonhome edge, so proposer DA matches all
home pairs.
\end{proof}

An active-screen stopping rule can exclude dormant home edges from execution
and make the assignment order dependent, without deleting their records.
The optional settlement in Corollary~\ref{cor:offscreen-settlement} instead
clears the full ledger. This construction cannot establish hardness of
complete ordering for that extension. The paper makes no complexity claim for the general
mandatory-complete-order problem or for ordering with expiration, revocation,
order-dependent responses, or dropout.

\section{Proofs and Extensions for Assignment-Aware Discovery}
\label{app:endogenous-recommendation-proofs}

\subsection{Adaptive attention gap}

\begin{proof}[Proof of Theorem~\ref{thm:adaptive-attention-gap}]
Proposition~\ref{prop:next-best-sufficiency} shows pathwise that correct
next-best responses reproduce the proposal trace of full-list proposer DA.
Because every pair is acceptable and the market is balanced, the terminal
matching is complete. The proposal bound is Wilson's coupon-collector bound,
and its asymptotic sharpness is due to Knuth
\citep{wilson1972analysis,knuth1976stable}. Equation
\eqref{eq:static-stable-deficit} and the nonadaptive upper bound are Theorems
5.3 and 5.2 of \citet{pittel2025constrained}.
\end{proof}

\subsection{Optimal adaptive discovery and implementation}

This benchmark takes rankings as known and treats compatibility as hidden;
the next-best benchmark instead assumes an exact preference oracle. They price
different information tasks. Let \(A\) be a public primitive graph with strict
endpoint rankings. Each edge
\(e\in A\) has an independently hidden compatibility bit
\(Z_e\sim\operatorname{Bernoulli}(p_e)\), where \(p_e\in(0,1)\). Querying that
bit costs \(c_e>0\). Let \(A_Z=\{e\in A:Z_e=1\}\) and define the target
\(\mu^\star(Z)=\operatorname{DA}_M(A_Z)\). The platform must return this target
for every realization. Let \(V(s)\) be the minimum expected remaining query
cost from transcript state \(s\) among zero-error policies. Verified compatible
pairs carry persistent execution authority; mere retention of a dormant
authorization would not meet this benchmark's premise.

At any transcript, let \(P\) be the verified-compatible edges and \(U\) the
unresolved edges. An edge in \(P\cup U\) is \emph{mutual first} when each
endpoint ranks the other above every other remaining potential partner. The
\emph{potential-core policy} repeatedly fixes a verified-compatible mutual-
first pair, queries an unresolved mutual-first pair, factors agent-disjoint
components, and stops when every completion of the unresolved bits has the same
DA matching. If none of those reductions applies, it queries an unresolved
edge minimizing
\begin{equation}
 c_e+(1-p_e)V(s_e^0)+p_eV(s_e^1),
 \label{eq:potential-core-bellman}
\end{equation}
and repeats the reductions in each answer branch. This Bellman step is needed
only inside a connected, target-varying residual component with no mutual-first
edge; such a component contains a directed preference cycle.

\begin{theorem}
\label{thm:potential-core-attention-optimality}
In the finite hidden-compatibility market above, the potential-core policy
terminates, returns \(\mu^\star(Z)\) for every compatibility realization, and
minimizes expected total query cost among all zero-error adaptive one-edge
query policies. If \(P_T\) is its terminal graph of verified compatible edges,
then
\[
  \operatorname{DA}_M(P_T)=\mu^\star(Z).
\]
Consequently, potential-core discovery followed by the optional query-neutral
settlement in Corollary~\ref{cor:offscreen-settlement} is
attention optimal among all zero-error joint discovery--execution controllers
using the same query language.
\end{theorem}

The optimum is over this finite query language with known rankings, not over
all preference-learning or recommendation technologies. Exact Bellman search
is a benchmark, not a claim of computational scalability. Query-neutral
settlement preserves an attained discovery cost; it does not remove the cost
of obtaining information outside the benchmark's premises.

\begin{proof}[Proof of
Theorem~\ref{thm:potential-core-attention-optimality}]
We use induction on the number of unresolved bits. First consider an unresolved
mutual-first edge \(e=(m,w)\). If \(Z_e=1\), then every stable matching of the
residual market contains \(e\): otherwise both endpoints are unmatched or
matched to lower-ranked partners and \(e\) blocks. If \(Z_e=0\), the pair is
infeasible. Holding all other bits fixed while switching \(Z_e\) therefore
changes \(\mu^\star(Z)\), so every zero-error decision tree must query \(e\) on
every root-to-leaf path. Move that unavoidable query to the root and, in each
answer branch, follow the original tree with its later query of \(e\) deleted.
This preserves every other query and weakly lowers pathwise cost. A positive
answer fixes the pair and removes its endpoints; a negative answer deletes the
edge. These are operations 1 and 2 of the policy.

Next suppose the residual market separates into agent-disjoint components.
Proposer DA and its target then factor component by component, and the hidden
bits are independent across components. For a lower bound, fix all states
outside one component in any joint zero-error tree. The induced local tree must
identify that component's target, so its conditional expected local cost is at
least the component value. Summing over components gives the sum of their
values. Running optimal component trees sequentially attains that sum, proving
the direct-sum reduction in operation 3.

If the target is constant over all completions, operation 4 is correct and no
query can improve on its zero continuation cost. Otherwise, after the preceding
reductions, every legal first action is a query of some \(e\in U\). Independence
leaves conditional answer probabilities \(1-p_e\) and \(p_e\); hence its value
is exactly
\[
 c_e+(1-p_e)V(s_e^0)+p_eV(s_e^1).
\]
Minimizing over \(e\) is therefore the Bellman equation in
\eqref{eq:potential-core-bellman}. Each answer removes one unresolved bit, so
the induction proves finite termination, zero-error correctness, and expected-
cost optimality. For interpretation, orient each endpoint toward its favorite
remaining incident edge. In a connected residual component with no mutual-first
edge, finiteness forces a directed preference cycle; this is precisely the core
left to the Bellman step.

At a terminal transcript, setting every unresolved bit to zero is one admissible
completion. Because all completions have the same target, DA on the verified-
compatible graph \(P_T\) equals \(\mu^\star(Z)\). The final implementation claim
then follows from Theorems~\ref{thm:main} and
\ref{thm:attention-execution-separation}.
\end{proof}

\subsection{Finite computational audit}

The exact solver is audited independently of the proof. The script
\texttt{audit\_potential\_core\_bellman.py} checks every hidden-compatibility
realization against ordinary proposer DA and compares the reduced recursion
with unrestricted Bellman search. Across all 16 two-by-two strict-preference
profiles, 14 are resolved by forced-pair and component reductions without cycle
search; the two crossed profiles
reach the residual cycle and require three expected queries when every bit has
probability one half. In 24 seeded three-by-three profiles, 14 reach a
cycle-Bellman state. The reduced solver visits 3,313 states on average and at
most 17,937, compared with 19,683 nodes in the full transcript DAG. For two
independent crossed two-by-two blocks, component factorization visits 175
states rather than the 6,561-node complete transcript DAG. A single run of the
full audit took 12.5 seconds on an Apple M1 laptop.

These counts verify that the structural reductions are real and expose where
the exponential work occurs. They are not a scaling theorem: three-by-three
markets are small, wall-clock time is machine specific, and the paper does not
assume that the residual cycle core remains bounded in large applications.

\subsection{Blocker-directed recommendation}

\begin{proof}[Proof of Theorem~\ref{thm:endogenous-recommendation}]
If \(x_{mt}\neq\mu_t(m)\), certification implies
\(x_{mt}\succ_m\mu_t(m)\), while membership in the challenger menu implies
\(m\succ_{x_{mt}}\mu_t(x_{mt})\). The new edge is therefore a primitive
blocking pair. Since it was omitted from \(G_t\), graph memory grows strictly
at every nonterminal iteration and finiteness gives the iteration bound.

At termination there is no block inside \(G_T\), because
\(\mu^T=\operatorname{DA}_M(G_T)\). If an omitted pair blocked, its receiver
would belong to
the corresponding willing-challenger menu and its proposer would prefer that
receiver to his assignment. The certified favorite could not equal the
current assignment, contradicting termination. Full stability follows. A
Pareto improvement would contain a new edge strictly preferred by both
endpoints and hence a blocking pair, proving efficiency.
\end{proof}

\section{Application Identification and Supplementary Results}
\label{app:application-supplement}

\subsection{Proof of the replay-identification result}

\begin{proof}[Proof of Proposition~\ref{prop:complete-profile-replay}]
Fix an event and a policy seed. At each history, the policy selects a recorded
pair; Assumption~\ref{ass:empirical-policy-invariance} supplies that pair's
invariant ranking and acceptability response. Recursion therefore identifies
the terminal state and every stated outcome. Averaging over the known seed
distribution gives the stochastic-policy result.
\end{proof}

\subsection{Sample and opportunity-budget controls}

The raw file has 8,378 directed reports, of which 8,368 have usable partner
identifiers and 8,128 also have like scores. Reciprocal rating completeness
leaves 8,022 directed reports in 4,011 dyads. Of these, 686 have mutual
affirmative decisions; four additional mutual-yes dyads lack complete ratings
and are excluded. The 277 proposers and 274 receivers supply 268 short-side
places across events. Eight events are unbalanced. Seventeen observed graphs
omit at least one pair, so full observed support is not the complete market.
Eight zero like scores are retained as recorded, below positive scores,
despite the questionnaire's nominal 1--10 range. Excluding those observations
changes the pooled position slope from \(-0.154\) to \(-0.145\).

The codebook reports a cap of affirmative decisions for half the partners in
event 12. Excluding that event leaves mutual-consent match rates of 60.58
percent on the full DA path, 51.22 percent under random DA, and 56.69 percent
under same-graph maximum matching. These are equal-weight means over the
remaining 20 events using the original event-level simulations. Thus the
main comparison is not driven by the capped event; its exclusion does not
identify how consent would respond to a different exposure policy.

The main experiment fixes the total number of opportunities in each event.
To separate their allocation across proposers from partner selection, we also
hold each proposer's original DA-path degree fixed and sample that many
executable partners uniformly without replacement. Each of 1,000 seeded draws
is cleared by both DA and maximum matching. The comparison still conditions
on an ex post degree vector and is not a deployable equal-delivery-cost policy.

\begin{table}[!htbp]
\centering\small
\caption{Controlling each proposer's opportunity budget}
\label{tab:proposer-degree-control}
\begin{tabular}{@{}llrrr@{}}
\toprule
Graph & Fixed budget & Random DA & Maximum & Path advantage \\
 & & (\%) & (\%) & (pp) \\
\midrule
Mutual consent & Total edges & 50.96 & 56.23 & 9.80 \\
 & Each proposer & 57.36 & 59.90 & 3.39 \\
All rated partners & Total edges & 85.88 & 95.30 & 14.12 \\
 & Each proposer & 84.84 & 94.06 & 15.16 \\
\bottomrule
\end{tabular}
\begin{minipage}{0.96\textwidth}\footnotesize
\emph{Notes:} Rankings and primitive consent graphs are fixed. Each row averages
over draws within event and then equally across 21 events. Maximum matching
and DA use identical edges within each draw. Path advantage compares
full-graph DA with the row's random DA. The per-proposer control fixes both
total opportunities and their allocation on the proposing side, not receiver
degrees. Its mutual-consent path advantage has event-level standard error
1.27 points; uncertainty is conditional on the replay model.
\end{minipage}
\end{table}

\subsection{Ordinal sensitivity to encounter-position drift}
The position regression concerns cardinal report levels. A common position
component in a report need not change the underlying ordering of partners;
nor does a small average coefficient exclude heterogeneous sequence effects.
We therefore retain replay invariance as an assumption and use detrending as
a conditional sensitivity exercise, not as an identified correction.

For each directed report with a usable like score, define
\[
  \widetilde R_{ije}=R_{ije}-\widehat\rho_{g(i),-e}(P_{ije}-1/2).
\]
The pooled specification replaces the subscripted slope by the pooled
evaluator-and-partner-fixed-effect estimate, $-0.153690$. The
leave-one-event-out specification estimates separate proposer and receiver
slopes after excluding the entire evaluation event. Only the position
component is removed: evaluator and partner fixed effects are not subtracted
from the scores. The original attractiveness and identifier tie breakers are
then applied. Decisions and primitive edges are unchanged.

For each event and graph definition, $Q_e^0$ is the full-graph DA proposal
count under the original rankings. Each of 1,000 uniformly sampled graphs
contains exactly $Q_e^0$ edges, and the very same sampled graph is supplied
to every ranking specification. Both DA and maximum-cardinality matching are
evaluated on that graph. The event budget is not recalibrated after
detrending, so changes in full-graph proposal counts cannot change the
comparison's cost. The table gives equal-weight means across the 21 events
after averaging over draws within event. Maximum-cardinality rates are
unchanged by construction because only rankings change.

Among the 2,446 within-person comparisons between executable mutual-consent
partners, 702 compare equal like scores. Pooled detrending changes 323 of
these tied-score comparisons and none of the 1,744 unequal-score comparisons.
The side-specific leave-one-event-out correction gives the same counts.
In the all-partners graph, pooled detrending changes 6,737 of the 13,604
tied-score comparisons and none of the 47,267 unequal-score comparisons.
The replication also reports the original pooled slope interval's endpoints
and a small $-0.01$ first-to-last drift. These scenarios are sensitivity
calibrations, not a confidence set for a policy effect. The exercise holds
consent fixed and does not test its invariance.

The replication provides the slopes, event-level outcomes, reversal
denominators, Monte Carlo diagnostics, seed construction, and input hashes.
It reproduces the original equal-budget DA and maximum-match rates before
applying any correction. No alternative oracle protocol is used in this
table.

\begin{table}[!htbp]
\centering
\small
\caption{Matching sensitivity to encounter-position detrending}
\label{tab:rank-detrending-sensitivity}
\begin{tabular}{@{}lrrr@{}}
\hline\hline
Ranking specification & Random DA & Same-graph maximum & Stability gap \\
 & (\%) & (\%) & (pp) \\
\hline
\multicolumn{4}{@{}l}{\emph{A. Mutual consent}} \\
Original rankings & 50.96 & 56.23 & 5.27 \\
Pooled position detrending & 51.61 & 56.23 & 4.62 \\
Leave-one-event-out, by side & 51.61 & 56.23 & 4.62 \\
Small drift: $\rho=-0.01$ & 51.58 & 56.23 & 4.65 \\
\hline
\multicolumn{4}{@{}l}{\emph{B. All rated partners}} \\
Original rankings & 85.88 & 95.30 & 9.42 \\
Pooled position detrending & 85.75 & 95.30 & 9.55 \\
Leave-one-event-out, by side & 85.75 & 95.30 & 9.55 \\
Small drift: $\rho=-0.01$ & 85.74 & 95.30 & 9.56 \\
\hline
\end{tabular}
\begin{minipage}{0.96\textwidth}
\footnotesize\emph{Notes:} For each event and opportunity definition, the edge budget is the full-graph DA proposal count under the original rankings. The budget, the 1,000 sampled graphs, and the bilateral consent decisions are held fixed across ranking specifications. DA and maximum matching use exactly the same graph. Entries average first over draws and then equally over 21 events. Detrending subtracts the estimated normalized-position component from like scores before applying the original attractiveness and identifier tie breakers. The leave-one-event-out specification estimates separate proposer and receiver slopes outside the evaluation event. The small-drift row uses a first-to-last score change of $-0.01$. This table isolates ranking sensitivity; it does not estimate a change in consent or a deployment effect.
\end{minipage}
\end{table}

\section{Related Literature and the Scope of the Contribution}
\label{app:literature-design}

The positioning claim in the introduction concerns the interaction of scarce
attention, persistent applications, and stable clearing. Earlier work supplies
the allocation rules, random-market limits, search models, and elicitation
techniques. Table
\ref{tab:joint-design-literature} groups the closest published work by the
objects it makes endogenous.  It does not claim that recommendation and
matching have never been combined.  It shows instead that the established
papers solve different information and implementation problems. Several
already combine adaptive learning with stable clearing; the comparison turns
on the state and constraints, not the presence of multiple stages.

\clearpage
\begin{table}[H]
\centering
\caption{Discovery, persistent state, and clearing in the closest literatures}
\label{tab:joint-design-literature}
\small
\setlength{\tabcolsep}{4pt}
\renewcommand{\arraystretch}{1.15}
\begin{tabular}{>{\raggedright\arraybackslash}p{0.18\textwidth}
                >{\raggedright\arraybackslash}p{0.22\textwidth}
                >{\raggedright\arraybackslash}p{0.22\textwidth}
                >{\raggedright\arraybackslash}p{0.27\textwidth}}
\toprule
Literature & Discovery or recommendation & Persistent state & Allocation object \\
\midrule
Attention in product and platform markets
\citep{eliaz2011consideration,declippeleliazrozen2014inattention,
cusumano2024competing,pratvalletti2022attention,teh2024accessibility,
chen2026attention}
& Consumers allocate inspection; firms and platforms influence visibility,
marketing, prices, or market access.
& Information and awareness may persist, but there is no off-screen bilateral
authorization ledger.
& Consumer choice, trade, prices, or directed-search efficiency; no exclusive
stable allocation. \\

Limited consideration
\citep{masatlioglu2012revealed,manzini2014stochastic,cattaneo2020random}
& Attention is modeled or identified from menu choice.
& No bilateral opportunity state.
& No rival-partner allocation or clearing rule. \\

Directed search and assortment
\citep{kanoria2021facilitating,immorlica2023designing,ashlagi2022assortment,
rios2023assortment,rios2026dating,su2022rankings}
& The platform chooses visibility, meeting rates, menus, or rankings.
& Some dynamic models retain likes or behavioral history, but not an authorized
graph for stable reclearing.
& Decentralized applications or reciprocal likes; objective is match volume or
welfare rather than stability. \\

Information acquisition and stable clearing
\citep{ashlagi2020clearing,arteaga2022smart,allman2023interviewing,
stephenson2022assignment,agarwal2023stable,grenet2022preference}
& Signals, search guidance, interviews, or assignment feedback determine the
information supplied to a fixed matching rule.
& Information or the interview graph carries into final clearing.
& Stable matching is reached through information acquisition, which can be
interleaved with clearing. Capacity-driven dormancy of valid applications is
the distinct implementation restriction studied here. \\

Prediction and repeated matching
\citep{ionescu2023strategic}
& Predicted outcomes help form preferences.
& Match outcomes update prediction across cohorts.
& A generic matching mechanism assigns each cohort; scarce within-market
attention and terminal stability are not the object. \\

This paper
& Assignment-aware discovery creates opportunities; a same-graph benchmark
separates their scarcity from the additional cost of stability.
& Valid applications survive off screen and reengage when eligible again.
& Online LA-DA implements proposer DA on final consideration; clearing the
entire retained ledger requires extra execution rights. A short TTC
extension considers the corresponding exchange problem. \\
\bottomrule
\end{tabular}
\end{table}

The first row supplies the broader economic framing. In
\citet{declippeleliazrozen2014inattention}, a consumer can inspect at most a
fixed number of product markets, and firms adjust prices because attention to
one market displaces attention from another. In
\citet{teh2024accessibility}, a platform chooses a meeting technology that
limits which buyers observe each seller, and imperfect accessibility can be
efficient in directed-search equilibrium. Platform attention is also traded or
allocated in \citet{pratvalletti2022attention} and \citet{chen2026attention}.
These results establish that limited capacity and endogenous visibility are
important market primitives. The present paper adds the no-blocking
requirement of exclusive matching and separates opportunity creation from
the execution of valid accumulated applications. The TTC extension concerns
a different allocation criterion and does not establish a common resource
frontier across mechanisms.

The comparison also separates two objectives that are sometimes conflated.
The market-wide ranking problem of \citet{su2022rankings} internalizes
competition for receiver capacity, but its welfare-maximizing ranking need not
be stable.  The dynamic dating policies of \citet{rios2023assortment} and
\citet{rios2026dating} integrate exposure, sequencing, and interaction history,
but optimize expected matches.  Conversely, the recommendation protocol of
\citet{ashlagi2020clearing} obtains a stable matching with low communication,
and updates qualification thresholds as applications arrive during DA.
Assignment-responsive learning is therefore an established design margin.
\citet{shi2022optimal} compares matchmaking architectures through their
communication requirements. Our distinction is the separation of a bounded
current display from valid accumulated authority, not the first integration
of recommendation and matching.

Persistence also has direct predecessors. The cumulative-offer process of
\citet{hatfield2005contracts} lets receivers choose from accumulated offers;
remembering applications is not itself new. The difference here is the
separate capacity-constrained consideration set, which can make an edge
ineligible while its authorization remains valid. Online LA-DA respects that
eligibility constraint and reengages the record when the pair becomes eligible
again. A cumulative-offer implementation augmented with the same eligibility
updates and reengagement is a legitimate comparator. Clearing every recorded
offer regardless of eligibility instead uses the additional rights in
Corollary~\ref{cor:offscreen-settlement}.

There is a second close comparison beyond matching protocols.
\citet{mackenzie2022menu} study dynamic menu mechanisms and strategic
implementation across several allocation environments. Their generality
is a foundation for expressing the original mechanism as sequential choices,
not a property first established here. \citet{peters2022elicitation} studies
partial-preference elicitation of necessarily optimal matchings, with explicit
query models and competitive guarantees. Identifying an allocation without
recovering every rank is therefore also established. Our construction uses
small attended comparisons in the TTC extension and keeps valid execution
authority separate from that evidence.
Neither binary tournaments nor caching alone is the claimed economic object.
The theorem's truthful-response premise is not a substitute for the incentive
conditions studied in strategic implementation.

Housing TTC and school-choice TTC also have different economic domains.
The endowed-object core theorem is due to
\citet{shapley1974cores,roth1977domination}; the capacity-and-priority
school mechanism is due to \citet{abdulkadiroglu2003school}.
LA-TTC preserves the selected reference rule rather than changing its welfare
criterion. The appendix school-choice replay uses the public experimental setting
of \citet{chen2006school} to measure elicitation at an unchanged allocation,
not to claim a new experimental treatment.

\citet{grenet2022preference} are the closest institutional bridge. Germany's
DoSV implements early Gale--Shapley stages in real time, lets students retain
multiple nonexpiring offers, and then completes centralized allocation. They
show quasi-experimentally that offer timing changes costly preference discovery
and propose batching offers. Our fixed-preference theorem asks a complementary
question: when current attention capacity changes after an application has
already been authorized, what state must survive and how can that accumulated
state be cleared exactly? Their event-level application, offer, revision, and
final-assignment records would permit the two margins---preference learning and
persistent authorization---to be studied jointly.

For the opportunity frontier, \citet[Theorem 11]{arnosti2015short} already gives
the stable-side fixed point used here. Theorem
\ref{thm:attention-cost-stability} combines that result with the Poisson
specialization of \citet{bordenave2013matchings}. Two further papers extend the
random-market analysis in economically important directions.
\citet{kanoria2025competition} vary connectivity and small market imbalance to
characterize competition, welfare, and unmatched agents under stability.
\citet{potukuchi2025unbalanced} give a sharp partial-list threshold for perfect
stable matching in unbalanced markets. Theorem
\ref{thm:attention-cost-stability} instead holds the realized graph fixed and
uses maximum matching on that same graph as the feasibility counterfactual.
This isolates the stability tax but does not extend the two papers' imbalance
results.

\citet{agarwal2023stable} ask which proposals should be made and construct
short lists that retain stable partners under correlated utilities. Their
lists are selected through an initial communication phase. Our budgeted
planning result makes a narrower point: heterogeneous delivery costs already
embed Knapsack before matching interactions arise. It does not claim that a
mandatory complete receiver order is hard; cumulative authorization can make
such an order allocation irrelevant in the sparse baseline. By contrast,
\citet{kanoria2021facilitating} endogenize costly strategic search and platform
restrictions in a decentralized dynamic equilibrium. Our adaptive controller
directs queries from the current tentative assignment and uses LA-DA to
implement a certified stable outcome. These differences identify the paper's
contribution without claiming that earlier work separates search from matching
altogether.

\section{External Evidence on the Design Margins}
\label{app:external-evidence}

This appendix reports construction and robustness details for the Chilean
exercise and benchmarks the stages of the attention funnel in other settings.
These sources document sparse consideration, bilateral conversion, feedback,
and retained reengagement as operational margins, not a LA-DA treatment effect.

\subsection{School choice and retained reengagement}

Section~\ref{sec:platform-evidence} uses Chile's public 2024 regular-stage SAE
files \citep{mineduc2025sae}. We exclude schools automatically added to preserve
enrollment continuity when measuring submitted list length. The local menu
consists of schools reporting an initial vacancy that are within 5 km of the
applicant's noisy geocode and compatible with the applicant's grade and the
school's gender restriction. Menu medians condition on 428,488 applicants
with usable coordinates and a positive local menu; 11,419 geocoded applicants
with empty menus are excluded. It is an institutional opportunity benchmark,
not a set of acceptable or preferred omitted schools.

The retained sample contains applicants who accepted an initial placement and
elected to wait for a better submitted choice, plus initially unassigned
applicants who remained in the process. An accepted gain is either a move to a
different school at a strictly higher-ranked course option, a move from an
automatic continuity course to a voluntarily ranked course at another school,
or a newly accepted ranked placement for an initially unassigned applicant.
We join each assignment to its actual course, not the applicant's minimum
rank at that school. Retained options count distinct other schools with a
voluntary course above that assignment, or all other voluntarily ranked
schools when the initial course is automatic continuity. Within-school course
improvements are excluded. Acceptance codes establish recorded acceptance, not
whether a fresh user action was needed: every applicant who initially accepted
and waited has code 1 at the later stage, including unchanged placements.
The public family-link file
identifies 12,173 retained-stage applicants in administrative family blocks.
Excluding them leaves 8,644 accepted gains among 105,236 applicants, an 8.2
percent rate compared with 8.1 percent in the full retained sample.

The separate wait-list capacity file changes reported vacancies for only 34 of
79,168 course records: 22 expansions add 22 vacancies and 12 contractions
remove 150. These revisions are neither large enough nor plausibly exogenous
enough to support a natural-experiment design. The exercise therefore documents
the use and consequences of persistent rankings; it does not identify the
assignment a system without retained rankings would have produced.

Germany's DoSV provides a sharper prospective setting for separating state
persistence from preference formation. \citet{grenet2022preference} observe 34
days of changing rank-order lists, program offers, student retention and
acceptance decisions, and final assignment. Offers do not expire during the
dynamic phase. Their quasi-experimental evidence shows that early offers alter
students' information acquisition and acceptance, so a fixed-preference replay
would be inappropriate. Event histories could instead measure which retained
offers become pivotal while modeling preference discovery jointly. We do not
analyze their restricted data here.

A randomized laboratory result isolates the execution margin.
\citet{stephenson2022assignment} finds that real-time provisional-assignment
feedback in DA sessions raises equilibrium assignments from 69.1 to 98.8
percent and the share without justified envy from 52.1 to 95.4 percent, while
truthful reporting is essentially unchanged. With three sessions per arm,
the two-sided exact randomization \(p\)-value is 0.10 for each allocation
outcome and 1.00 for truthful reporting. The randomized contrast concerns
feedback, but six sessions provide limited precision; unchanged reporting
is not an equivalence finding.

\subsection{Dating and physician funnels}

Physician matching displays the attention funnel at the application stage. In
the Canadian residency market, applications per applicant rise from 21.5 in
2020 to 27.8 in 2025, while the application-level interview-offer rate falls
from 43.6 to 37.5 percent and the share obtaining any interview rises from 81.5
to 90.5 percent \citep{carms2025interviews}. These aggregates do not identify
congestion or reveal complete rankings. They motivate linked data separating
application, interview, rank, and match stages.

Two dating datasets distinguish profile delivery from executable attention.
The curated-dating records of \citet{eastwick2025young} contain 4,542 arranged
dates. Among the 4,532 dates with two reported decisions, 35.3 percent produce mutual
interest in another date and 44.5 percent produce interest from only one side.
The eHarmony archive of \citet{dinh2022computational} contains 127,699 dated
algorithm-delivered dyads. Within user and delivery date, one additional
recommendation in the preceding six days is associated with a 0.375
percentage-point reduction in per-profile communication initiation. This
regression equally weights 191,834 user-date observations from 50,983 users
observed on at least two dates, with user and calendar-date fixed effects.
A negative
future-load placebo indicates selected timing, so this gradient is descriptive
rather than causal.

Together, the sources map delivery to processing, processing to bilateral
consent, and executable opportunity to stable assignment. Only the speed-dating
data reveal bilateral rankings throughout the primitive dating market needed
for the paper's structural matching counterfactual. Chile records submitted
rankings but not preferences over omitted schools; the remaining sources
identify earlier conversion or reengagement margins.

\section{Limited-Attention Exchange}
\label{app:la-ttc-extension}

\subsection{Exactness and allocation properties}

For Proposition~\ref{prop:la-ttc-extension}, write
\(B_i=\{j:(i,j)\in G\}\) and let \(R\) be the set of remaining owners.
Owner \(i\)'s choice menu is \(B_i\cap R\), which contains \(i\) until
that owner leaves. Enumerate this finite menu independently of unobserved
tastes. A binary tournament retains an incumbent and compares it with the
next alternative. After each comparison, the incumbent is best among the
alternatives processed so far. Strict transitivity establishes the true best
element of the menu. Reusing a sound recorded comparison leaves this
invariant unchanged. A singleton requires no query.

Thus every LA-TTC pointer agrees with the corresponding reference TTC
pointer. A finite total directed graph contains a cycle. Removing its owners
and endowed objects leaves precisely the reference residual market, and
induction gives the same allocation. At least one owner leaves each round,
so there are at most \(|N|\) rounds. If \(d_i=|B_i|\), there are at most
\(\sum_i\binom{d_i}{2}\) new successful comparisons. This bound excludes
failed contacts, elapsed time, and permission acquisition. It is not a
minimum-attention claim.

To apply the classical housing result, extend each owner's ranking by
placing every object outside \(B_i\) below the endowment, while preserving
the order on \(B_i\). Because the endowment is always available until its
owner leaves, ordinary TTC never selects an excluded object. The extended
and restricted procedures coincide. Classical TTC is individually rational
and in the strong core \citep{shapley1974cores,roth1977domination}. A
coalition that could weakly improve every member and strictly improve one
using its own endowed objects and \(G\)-eligible assignments would also
block the extended allocation. No such coalition exists. Taking the grand
coalition yields Pareto efficiency on \(G\).

For settlement, a trading owner's permission must authorize both receipt of
the assigned object and release of the endowment within the complete
exchange. A standing mandate can supply this permission earlier. Refusal,
expiry, or withdrawal prevents the corresponding transfer; it is not
overridden by the stored ranking. The institution may instead update the
reference economy and compute a new allocation.

\subsection{Why transfers must wait}

Three owners have fixed preferences
\[
 0:1\succ0\succ2,\qquad
 1:2\succ0\succ1,\qquad
 2:0\succ2\succ1.
\]
Self-edges and \((0,1),(1,0)\) support TTC allocation \((1,0,2)\).
Adding \((1,2),(2,0)\) changes it to \((1,2,0)\), leaving owner 0's
assignment unchanged and strictly improving owners 1 and 2. If the initial
trades are irreversible and the participants exit, retaining their earlier
decisions cannot implement the later reference allocation. This is not a
failure of TTC on either fixed graph.

A sufficient early-execution condition is available. Suppose
\(L\subseteq G\subseteq U\) bounds every admissible final graph. If a
cycle uses guaranteed edges in \(L\), its permissions remain valid, and
each pointed-to object is best among all remaining possibilities in \(U\),
the same cycle occurs in every completion. Standard TTC cycle-removal order
invariance permits removing it first. A record of past consent is not a
guarantee against actual withdrawal.

\subsection{Formal-verification scope}

The comparison construction and concrete housing TTC program use these
files and declarations:
\begin{verbatim}
ChoiceMechanismCompiler.lean
ChoiceDrivenTTC.lean:
  program_run
  tree_exact
  tree_correct_and_bounded
\end{verbatim}
They verify the connection to the cycle kernel, allocation exactness, and a
finite comparison bound. The sharper
graph-dependent bound and strong-core transfer above are mathematical
proofs. Valid exchange authority is an economic input, not inferred by Lean
from a comparison. The school-capacity implementation below has separate
mathematical arguments and executable tests; the housing core theorem is
not a school-priority stability theorem.

\subsection{A school-choice query comparison}
\label{sec:school-elicitation}

The public experimental data of \citet{chen2006school} permit a same-allocation
comparison of preference-query procedures. The recorded capacities,
priorities, lotteries, submitted rankings, and allocations reproduce all
432 participant-session assignments without a mismatch. The counterfactual
uses the experiment's induced strict preferences to simulate truthful
binary responses. The learner is not supplied with those complete rankings.

Both procedures retain comparisons and transitive implications. One requests
the best choices reached by the reference school DA or TTC rule. The other
interleaves merge sort across students but stops as soon as the same
sufficient allocation certificate succeeds. Neither procedure is forced to
learn a full ranking. Both know initially that every school is preferable
to the outside option. The comparison changes elicitation order, not the
allocation rule or permission regime.

\begin{table}[H]
\centering
\caption{Comparisons needed to recover the same school allocation}
\label{tab:school-elicitation}
\begin{tabular}{llrrr}
\toprule
Rule & Preference design & Rule-directed & Early-stop sorting & Full sorting \\
\midrule
TTC & Designed & 375.0 & 396.7 & 433.0 \\
TTC & Random & 241.7 & 377.3 & 462.0 \\
TTC & Pooled & 308.3 & 387.0 & 447.5 \\
DA & Designed & 437.3 & 407.7 & 433.0 \\
DA & Random & 279.5 & 415.7 & 462.0 \\
DA & Pooled & 358.4 & 411.7 & 447.5 \\
\bottomrule
\end{tabular}
\begin{minipage}{0.96\linewidth}
\footnotesize Notes: Mean distinct binary comparisons per 36-student,
seven-school market. Each preference design is evaluated at six recorded
lotteries. Every completed elicitation returns its reference allocation.
These are induced-response counterfactuals, not measured attention treatments.
\end{minipage}
\end{table}

Rule-directed TTC uses 20.3 percent fewer comparisons in the pooled sample
than early-stopping sorting, measured as the ratio of mean counts. The
ranking reverses for DA in the designed preference profile. The exercise
therefore illustrates attainable savings under a specified response model,
not uniform query dominance or an estimated effect of persistent authority.

\subsection{Replication and query accounting}
\label{app:school-elicitation-audit}

The source for Section~\ref{sec:school-elicitation} is the author-hosted
ASCII release for \citet{chen2006school}. It contains 432 participant-session
records, not 432 independently identified people: the original study notes
two repeat participants. The replication checks the source hash, within-session
identifiers, complete ranking permutations, induced values, school capacities,
lotteries, and allocations before computing the counterfactuals.

Every school places district students first and then uses the recorded lottery
position. The ordered lottery list is inverted before constructing priorities;
priorities are not fitted to observed assignments. Running the observed
Boston, DA, or TTC rule on submitted rankings reproduces all 432 recorded
assignments. The counterfactual then replaces submitted rankings with the
experiment's induced preferences and holds the institutional inputs fixed.
Both allocation rules are replayed on all twelve recorded lotteries, regardless
of the rule used in the original session.

The learner sees only a binary response interface and common prior knowledge
that every school is preferable to outside. The evaluator's access to full
induced rankings is used to simulate responses and check exactness, not to
choose the learner's questions. Both arms remember comparisons and their
transitive implications. The mechanism-directed arm requests only the choices
reached by its school DA or TTC computation. The sorting arm uses merge sort,
interleaved round-robin across students, and stops when the same sufficient
certificate succeeds. Its first certified prefix is checked at the boundary;
a fixture checks every prefix. Certificate failure is not an identification
lower bound against arbitrary alternative elicitation algorithms.

The twelve session means in Table~\ref{tab:school-elicitation} are six lottery
realizations in each of two fixed preference designs. The pooled TTC reduction
is the ratio of average query counts, not an average causal effect or an
estimate based on twelve independently sampled preference profiles. No
standard errors are assigned to these deterministic counterfactuals.
All complete runs reproduce the corresponding induced-preference allocation.
School TTC can have justified envy; the audit records it without imposing
DA's priority-stability criterion.

The aggregate output is stored in
\path{theory/computation/limited_attention_school_choice/chen_sonmez_2006_audit.json}.
From the repository root, the local source and query audit is reproduced by
\begin{verbatim}
python3 -B empirical/src/audit_school_choice.py --self-test --check
python3 -B -m unittest empirical.tests.test_limited_attention_school_choice
\end{verbatim}
The adjacent replication README supplies the public download URL and hash.
The raw release is not redistributed with the paper package. The audit uses
no fitted attention process and does not observe responses to LA-DA or
LA-TTC. Learning burden, actual processing, and authority to settle remain
distinct objects.

\end{document}